\RequirePackage[hyphens]{url}
\documentclass[submission,copyright,creativecommons]{eptcs}

\providecommand{\event}{GCM 2024} 

\usepackage{iftex}

\ifpdf
  \usepackage{underscore}         
  \usepackage[T1]{fontenc}        
\else
  \usepackage{breakurl}           
\fi

\title{Scalable Pattern Matching in Computation Graphs}
\author{Luca Mondada
\institute{University of Oxford\\Oxford, UK}
\institute{Quantinuum Ltd\\Cambridge, UK}
\email{luca.mondada@cs.ox.ac.uk}
\and
Pablo {Andrés-Martínez}
\institute{Quantinuum Ltd\\Cambridge, UK}
}

\def\titlerunning{Scalable Pattern Matching in Computation Graphs}
\def\authorrunning{L. Mondada and P. Andrés-Martínez}

\usepackage[svgnames]{xcolor}
\usepackage{tikz}
\usetikzlibrary{quantikz2}
\usetikzlibrary{graphs}
\usetikzlibrary{decorations.pathreplacing,calligraphy}
\usepackage{tikzit}

\tikzstyle{deg2-gate}=[fill=white, draw=black, shape=rectangle, minimum width=1.2cm, minimum height=1cm]
\tikzstyle{deg2-gate-grey}=[fill={rgb,255: red,216; green,216; blue,216}, draw=black, shape=rectangle, minimum width=1.2cm, minimum height=1cm]
\tikzstyle{deg1-gate}=[fill=white, draw=black, shape=rectangle, minimum width=1.2cm, minimum height=4mm]
\tikzstyle{deg1-gate-half}=[fill=white, draw=black, shape=rectangle, minimum width=6mm, minimum height=5mm]
\tikzstyle{red-rect}=[fill={rgb,255: red,255; green,184; blue,185}, draw={rgb,255: red,245; green,22; blue,99}, shape=rectangle, minimum width=1.2cm, minimum height=4mm]
\tikzstyle{purple-rect}=[fill=white, draw={rgb,255: red,129; green,89; blue,159}, shape=rectangle, minimum width=5cm, minimum height=6mm]
\tikzstyle{salmon-rect}=[fill=white, draw={rgb,255: red,222; green,179; blue,171}, shape=rectangle, minimum width=5cm, minimum height=6mm]
\tikzstyle{big-transparent}=[fill=white, draw=white, shape=rectangle, minimum width=8cm, minimum height=3cm, opacity=0.4]
\tikzstyle{deg3-gate}=[fill=white, draw=black, shape=rectangle, minimum height=15mm, minimum width=12mm]
\tikzstyle{deg2-gate-no-width}=[fill=white, draw=black, shape=rectangle, minimum height=5mm]
\tikzstyle{deg1-gate-grey}=[fill={rgb,255: red,216; green,216; blue,216}, draw=black, shape=rectangle, minimum height=4mm, minimum width=1.2cm]
\tikzstyle{red-rect-big}=[fill={rgb,255: red,255; green,184; blue,185}, draw={rgb,255: red,245; green,22; blue,99}, shape=rectangle, minimum width=1.2cm, minimum height=1cm]

\tikzstyle{edge1}=[-, draw={rgb,255: red,245; green,22; blue,99}, thick]
\tikzstyle{edge2}=[-, draw={rgb,255: red,8; green,35; blue,133}, thick]
\tikzstyle{edge3}=[-, draw={rgb,255: red,9; green,197; blue,6}, thick]
\tikzstyle{black}=[-, thick]
\tikzstyle{arrow}=[->, thick]
\tikzstyle{dashredZ}=[-, draw={rgb,255: red,85; green,85; blue,85}, dashed, very thick]
\tikzstyle{dottedzz}=[-, dotted, thick]
\tikzstyle{new edge style 0}=[->, draw={rgb,255: red,85; green,85; blue,85}, dotted, thick]

\usepackage{adjustbox}

\usepackage{caption}
\usepackage{subcaption}
\usepackage{graphicx}
\usepackage{svg}
\usepackage{listings}
\usepackage{xifthen}
\usepackage[shortlabels]{enumitem}
\usepackage{afterpage}
\usepackage{xfrac}
\usepackage{siunitx}

\usepackage{amsmath}
\usepackage{amssymb}
\usepackage{amsthm}
\usepackage{mathrsfs}
\usepackage{thmtools}
\usepackage{thm-restate}

\graphicspath{ {./imgs/} }

\PassOptionsToPackage{hyphens}{url}
\usepackage{hyperref}
\usepackage{cleveref}

\theoremstyle{remark}

\theoremstyle{definition}

\theoremstyle{plain}
\newtheorem{prop}{Proposition}
\theoremstyle{plain}

\theoremstyle{plain}
\newtheorem{coro}[prop]{Corollary}

\crefname{prop}{proposition}{propositions}

\makeatletter
\let\commentfullflexible\lst@column@fullflexible
\makeatother

\makeatletter
\renewcommand{\paragraph}{%
  \@startsection{paragraph}{4}%
  {\z@}{2ex \@plus 1ex \@minus .2ex}{-1em}%
  {\normalfont\normalsize\bfseries}%
}
\makeatother

\setlist[description]{
  font=\normalfont\itshape, 
  labelindent=\parindent, 
  leftmargin=* 
}

\begin{document}

\maketitle

\begin{abstract}
  Graph rewriting is a popular tool for the optimisation and modification
  of graph expressions in domains such as compilers, machine learning and
  quantum computing.
  The underlying data structures are often port graphs---graphs with labels at edge endpoints.
  A pre-requisite for graph rewriting is the ability to find graph patterns.
  We propose a new solution to pattern matching in port graphs. Its novelty
  lies in the use of a pre-computed data structure that makes
  the pattern matching runtime complexity independent of the number of patterns.
  This offers a significant advantage over existing solutions for use cases
  with large sets of small patterns.

  Our approach is particularly well-suited for quantum superoptimisation.
  We provide an implementation and benchmarks showing that our algorithm offers a 20x speedup
  over current implementations on a dataset of \num{10000} real world patterns
  describing quantum circuits.
\end{abstract}

\section{Introduction}
Optimisation of computation graphs is a long-standing problem in computer science
that is seeing renewed interest in the compiler~\cite{mlir},
machine learning (ML)~\cite{taso,computationgraph}
and quantum computing communities~\cite{quartz,qeso}.
In all of these domains, graphs encode computations that are either expensive to execute or that are
evaluated repeatedly over many iterations,
making graph optimisation a primary concern.

Domain-specific heuristics are the most common approach in compiler
optimisations~\cite{pytorch,TKET}---
a more flexible alternative are optimisation engines based on \emph{rewrite rules},
describing the allowable graph transformations~\cite{bonchiI,bonchiII}.
Given a computation graph as input,
we find a sequence of rewrite rules that transform the input
into a computation graph with minimal cost.
One successful approach in both ML and quantum computing has been to use automatically generated rules, scaling to using hundreds and even thousands of rules~\cite{quartz,taso,qeso}.

In the implementations cited above, pattern matching is carried out separately for each pattern, becoming a bottleneck for large rule sets.
We present an algorithm for pattern matching on computation graphs that uses a pre-computed
data structure to return all pattern matches in a single query.
The set of rewrite rules are directly encoded in this data structure:
after a one-time cost for construction,
pattern matching queries can be answered in running
time independent of the number of rewrite rules.

We provide a novel solution to pattern matching on port graphs~\cite{portgraph}
with a runtime complexity independent of the number of patterns.
As a trade-off, the runtime may be exponential
in the size of the patterns.
For pattern sizes of practical interest in quantum computing, however,
the resulting costs are manageable:
the exponential scaling is in the number of qubits
of the patterns, which is bounded by a single digit constant in relevant rewriting use cases~\cite{quartz}.

The solution we propose can be seen as an adaptation of Rete networks~\cite{reteforgy} to the special case of port graph pattern matching.
The additional structure obtained from restricting our considerations to graphs results in a simplified network design and crucially, 
allows us to derive worst-case asymptotic runtime bounds---overcoming a key limitation of Rete.
A similar problem is also studied in the context
of multiple-query optimisation for database queries~\cite{sellis,mqsubiso}, but has limited itself to developing caching
strategies and search heuristics for specific use cases.
Finally, using a pre-compiled data structure for pattern matching
was already proposed in~\cite{messmer}.
However, with a $n^{\Theta(m)}$ space complexity---%
$n$ is the input size and $m$ the pattern size---%
it is a poor candidate for pattern matching on large graphs,
even for small patterns.

\section{Paper overview}

Taking advantage of \emph{port labels} on graph data structures
leads to a speedup for pattern matching over the general case~\cite{Jiang2}. 
Port labels are data assigned to every endpoint of the edges of a graph,
such that the labels at every vertex are unique.
Such labels can for instance be assigned to processes with distinguishable
inputs and outputs:
a function that maps inputs $(x_1, \dots, x_n)$ to output $(y_1, \dots, y_m)$
can assign labels $i_1$ to $i_n$ and $o_1$ to $o_m$
to its incident edges in the computation graph.
The resulting data structure is a \emph{port graph}~\cite{portgraph}.


\paragraph{Main idea.}
For a set of $\ell$ pattern port graphs $P_1, \dots, P_\ell$ and a subject port graph $G$,
we solve the problem of finding all pattern embeddings $P_i \to G$.
We distinguish two separate stages during pattern matching:
\begin{description}
    \item[Pre-computation stage.] Compile the set of patterns into a data structure designed to speed up later queries.
    In the process, we select for every pattern $P_i$ an \emph{anchor set}, i.e. a subset $X_i$ of vertices in $P_i$.
    The input port graph $G$ is not required at this stage and, hence, this computation is done only once for a given collection of patterns $P_1, \dots, P_\ell$.
    \item[Fast pattern matching stage.] Given $G$, compute for each $P_i$
    all embeddings $P_i \to G$. This is
    achieved by enumerating all possible images $X$ in $G$ of pattern anchor sets $X_1, \dots, X_\ell$ and
    finding the subset of patterns $i$ for which the map $X_i \to X$
    can be extended to a valid pattern embedding of $P_i$ in $G$.
\end{description}
The pattern matching stage can be decomposed into
known problems by showing that embeddings
with fixed anchor sets can be equivalently seen as rooted tree embeddings.
The enumeration of valid choices of $X$ in $G$ then becomes a tree counting argument.
Meanwhile, the set of patterns that embed in $G$ for a fixed $X_i \to X$ is obtained from a pre-computed decision tree.
An overview is presented in \cref{fig:intro}.


\begin{figure}[t]
\resizebox{\textwidth}{!}{\input{intro-fig.tikz}}
\caption{Pattern matching on a port graph is reduced to the
problem of matching on trees.
A subset of vertices are chosen as anchor sets (left).
A neighbourhood of the anchors is extracted and represented as a tree (middle).
Finally, pattern matches are found by searching for matching subtrees (right).}
\label{fig:intro}
\end{figure}

    
\paragraph{Results and contributions.}
Our first major contribution is a pattern matching algorithm for port graphs 
with a runtime complexity bound independent of the number of patterns being matched, achieved
using a one-off pre-computation.
The main complexity result is expressed in terms of maximal pattern
\emph{width} $w$ and \emph{depth} $d$, two measures of pattern
size defined in \cref{subsec:pg}.
These are directly related to the tree representation illustrated in \cref{fig:intro}:
width is equal to the size of the anchor set (\cref{prop:cananchors})
and depth is at most twice the tree height (\cref{eq:ctsize}).
We assume bounded degree graphs (the complete list of assumptions is given in \cref{sec:assumptions})
and we use
the \emph{graph size} $|G|$ to refer to the number of vertices in $G$.

\begin{restatable}{thm}{mainthm}\label{thm:main}
    Let $P_1, \dots, P_\ell$ be patterns with at most width $w$
    and depth $d$.
    The pre-computation runs in time and space complexity
    \[
    O \left( (d\cdot \ell)^w \cdot \ell + \ell \cdot w^2 \cdot d \right).
    \]
    For any subject port graph $G$, the pre-computed prefix tree can be used
    to find all
    pattern embeddings $P_i \to G$ in time
    \begin{equation}\label{eq:mainruncompl}
    O \left( |G| \cdot \frac{c^w}{w^{\sfrac{1}{2}}} \cdot d \right)
    \end{equation}
    where $c = 6.75$ is a constant.
\end{restatable}

\noindent
The runtime complexity is dominated by an exponential scaling in maximal pattern width $w$.
Meanwhile, the advantage of our approach over matching one pattern at a time grows with the number of patterns $\ell$.
It is thus of particular interest for matching numerous small width patterns.

We illustrate this point by comparing our approach
to a standard algorithm that matches one pattern
at a time~\cite{Jiang2}, with 
runtime complexity $O(\ell \cdot |P| \cdot |G|)$.
Using $|P| \leq w\cdot d$ (shown in \cref{subsec:pg}) and
comparing to eq. (1), we thus
have a speedup in the regime $\Theta(c^w / w^{\sfrac32}) < \ell$.
On the other hand, $\ell$ is upper bounded
by the maximum number $N_{w, d}$ of patterns of bounded width and depth.
Using a crude lower-bound for $N_{w,d}$ derived in \cref{app:proofellbound},
we obtain a computational advantage for our approach when
\begin{equation}\label{eq:regime}
    \Theta\left(\frac{c^w}{w^{\sfrac32}}\right) < \ell < \left(\frac{w}{2e}\right)^{\Theta(w d)} \leq N_{w, d}.
\end{equation}
Our second major contribution is an efficient Rust library for port graph pattern matching\footnote{portmatching: \url{https://github.com/lmondada/portmatching}}.
We present benchmarks on a real world dataset of \num{10000} quantum circuits in \cref{sec:impl}, showing
a 20$\times$ speedup over a leading C++ implementation of pattern matching for quantum circuits.

\section{Preliminaries}\label{sec:defs}

\subsection{Definitions}\label{subsec:pg}

A port graph $G$ is a tuple $G := (V,E,\mathcal{P},\lambda)$ where $(V,E)$ is an undirected multigraph,
$\mathcal{P}$ is a finite set of \emph{port labels} and
$\lambda \colon V \times \mathcal{P} \rightharpoonup E \cup \{ \omega \}$
is a partial function that on its domain of definition
either assigns port labels to edges or marks them as open ports
using a specially dedicated symbol $\omega$.
The port graph is valid if $\lambda(v, p) = \lambda(v, p') = e \in E$ if and only if $e$ is an edge incident in $v$
and $p = p'$.
We then say that $e$ is attached to $v$ at port $p$.
The domain of definition of $\lambda(v, \cdot)$ is the set of port labels at vertex $v$, written $ports(v)$.
The degree of $v$ is $deg(v) = |ports(v)|$.
It will often be convenient to leave the definition of $\lambda$ implicit and 
for an edge $e = \lambda(v, p) = \lambda(v', p')$, to write it instead as the set $e = \{(v, p), (v', p')\}$.
Additionally, we may consider port graphs with labelled vertices, determined by \textit{vertex label maps} $V \to \mathcal{W}$,
where $\mathcal{W}$ is a set of labels.

\paragraph{Width, depth and linear paths.}
Fix a partition of port labels $\mathcal{P}$ into pairs of elements (and an additional singleton set if $|\mathcal{P}|$ is odd).
This defines an equivalence relation $\sim$ where $p \sim p'$ if $p$ and $p'$
are in the same partition.
The relation $\sim$ defines a partition of the edges of a port graph into paths. 
A linear path in a port graph $G$ with edges $E$ and port labels
in $\mathcal P$
is a path $P \subseteq E^\ast$ such that for every vertex $v$ in $G$ and ports $p, p' \in \mathcal{P}$
satisfying $p \sim p'$,
\begin{equation}\label{eq:linpath}
  \lambda(v, p) \in P \quad\implies\quad \lambda(v, p') \in P \cup \{\omega\}.
\end{equation}
From a single edge in $G$, a linear path can be constructed uniquely by repeatedly
appending edges to the path so that (\ref{eq:linpath}) is satisfied.
The linear path decomposition of $G$ is the unique partition of the edges of $G$ into linear paths.
The width $width(G)$ of $G$ is the number of linear paths in this decomposition.
The depth $depth(G)$ of $G$ is the length of the longest linear path in $G$.
Every edge is on exactly one linear path, while vertices may be on zero, one or several linear paths.
We have always $|G| \leq width(G) \cdot depth(G)$.

\paragraph{Patterns and embeddings.}
A pattern is a connected port graph. A pattern embedding (or just embedding) $\varphi: P \to G$
from a pattern $P = (V_P, E_P, \mathcal{P}, \lambda_P)$
to a subject port graph $G = (V, E, \mathcal{P}, \lambda)$, 
both with identical port label sets,
is given by an injective vertex map $\varphi_V: V_P \to V$
such that
$\lambda_P(v, p)$ is defined if and only if $\lambda(\varphi_V(v), p)$ is defined
and the edge map $\varphi_E: E_P \to E$ defined as
\begin{equation}\label{eq:emb}
\varphi_E(e) = 
\lambda(\varphi_V(v), p)\quad\textrm{for }(v, p) \in V \times \mathcal P\ \textrm{ s.t. }\ \lambda_P(v, p) = e\\
\end{equation}
is well-defined and injective.
Finally, if the pattern and port graphs have node labels $V_P \to \mathcal{W}$
and $V \to \mathcal{W}$, then we also require that pattern
embeddings preserve those.

\subsection{Simplifying assumptions}\label{sec:assumptions}
We list here a series of assumptions made throughout our argument.
They represent a restriction from the most general case but we do not
find that they restrict the usefulness of the result in practice.
We show in \cref{sec:quantumcirc} that these assumptions hold in the case of quantum circuits.
Moreover, as discussed in \cref{sec:impl}, none of these assumptions
are required for the implementation, and we have
not observed any impact on performance when lifting them in practice,
so we conjecture that these assumptions can be loosened and our results generalised.
\begin{enumerate}
    \item All graphs and patterns are of bounded maximum degree $\Delta$.
    \item There is no linear path that forms a cycle.
    \item All pattern embeddings $\varphi: P \to G$ must be \emph{convex}, i.e. for every subgraph $H \subseteq G$
          that contains the image of $P$, $\varphi(P) \subseteq H$, it holds that $width(P) \leq width(H).$
\end{enumerate}
Note also that \cref{eq:emb} requires that the degree
of a vertex $v$ in $P$ is also preserved $deg(v) = deg(\varphi(v))$.
Importantly, a pattern embedding may
map a vertex with open port $p$ to a vertex
in the subject graph that has an edge attached to port $p$.

We will further simplify the problem by making choices of presentation that do not imply any loss of generality.
First of all, we will assume that all vertices are on at most two linear paths (and thus in particular $\Delta = 4$).
Vertices on $k > 2$ linear paths can always be broken up into a composition of $k-1$ vertices,
each on two linear paths as follows:
\ctikzfig{breakup} 
This transformation leaves graph width unchanged but may multiply the
graph depth by up to a factor $\Delta$.
We can then fix the set of port labels to the set $\mathcal{P} = \{i_1, i_2, o_1, o_2\}$
with a total order $\leq$ on $\mathcal{P}$\footnote{Any total order will work, e.g. $i_1 \leq i_2 \leq o_1 \leq o_2$.}.
We fix the partition of $\mathcal{P}$ into pairs of elements given by
$i_k \sim o_k$ for $k \in \{1, 2\}$.
We also enforce that at every vertex $v$, the set of
ports $ports(v)$ is partitioned by $\sim$ into $\lfloor ports(v) / 2\rfloor$ pairs of elements and at most one singleton set.
This can always be achieved by relabelling vertex ports.
Finally, we assume that all patterns have the same width $w$ and depth $d$,
are connected port graphs and have at least 2 vertices.

Using these assumptions we can obtain the following notable bound on graph width.
\begin{restatable}{prop}{flatgraphwidth}\label{prop:flatgraphwidth}
    Let $G$ be a port graph with $n_\textrm{odd}$ vertices of odd degree
    and $n_\omega$ open ports.
    Then the graph width of $G$ is at most $\lfloor(n_\textrm{odd} + n_\omega) / 2\rfloor$.
\end{restatable}
\noindent
The proof is in the \cref{app:flatgraphwidth}.

\paragraph{Rooted paths ordering.}
The total order on $\mathcal{P}$ also induces a total order on the paths $e_1\cdots e_k \in E^\ast$
in $G$ that start in the same vertex $r \in e_1$:
the paths are equivalently described by a string in $\mathcal{P}^\ast$, the sequence of ports 
of $e_1, \dots, e_k$,
which we order using the lexicographical ordering on strings.
For every vertex $v$ in $G$ there is thus a unique smallest path from $r$ to $v$ in $G$ that
is invariant under isomoprhism of the underlying graph (i.e. relabelling of the vertex set).

\subsection{Quantum Circuits}\label{sec:quantumcirc}
We see pattern matching for quantum circuits as one of the
main applications of our results.
We therefore choose to introduce here the quantum circuit syntax as a motivation and illustration
of the port graph formalism.
Similar data structures are also in use in other parts of compilation science, variously referred to as
circuits, computation graphs or dataflow graphs.
Readers familiar with port graphs or uninterested in the application in quantum computing may skip directly to the
definitions of \cref{subsec:pg}.

The set of operations in a quantum circuit is called the \emph{gate set} of the computation and forms
the set of node weights of the port graph.
To every element of the gate set, called a \emph{gate type}, is associated a gate arity.
A gate with gate type of arity $n$ has $n$ incoming port labels $i_k$
and $n$ outgoing port labels $o_k$---by our assumptions we thus assume $1 \leq n \leq 2$.
Edges always connect outgoing to incoming labels, and the directed port graph that results
from these edge orientations is acyclic.
A quantum circuit has $q$ qubits if it has
$q$ outgoing and $q$ incoming open ports: the inputs and outputs of the circuit.

All assumptions made in \cref{sec:assumptions} can be easily verified for
quantum circuits and in fact will also apply to most computation graphs
more generally. Indeed:
\begin{enumerate}
    \item  Bounded degree is a direct consequence of a having a fixed set of gate types.
    \item For any directed acyclic computation, using port labels in $i_k$
for incoming ports and $o_k$ for outgoing labels will always
result in non-cyclic linear paths.
    \item Similarly, convexity is a natural restriction in the context
of rewriting systems for acyclic digraphs, as it ensures that the application 
of a rewrite rule does not introduce a cycle in the graph.
\end{enumerate}
Using the $i_k \sim o_k$ port label partition, linear paths in a quantum circuit correspond to the paths of gates along a single qubit.
For applications to quantum circuits, we can thus bound graph width and depth as follows:
\begin{restatable}{prop}{quantumwidth}\label{prop:quantumwidth}
    The port graph $G$ of a quantum circuit with $q$ qubit and at most $d$ gates on any one qubit
    has width $q$ and depth $d$.\qed
\end{restatable}
\noindent
Some further considerations on applying our work to quantum circuits are discussed in \cref{app:qc-pg}.

\subsection{Pattern Matching}
In our pattern matching task, given a subject port graph $G$ and a collection of port graphs $P_1, \dots, P_\ell$, we must find all pattern embeddings
\begin{equation}\label{eq:patternmatches}
\{\varphi: P_i \to G\ | \ 1 \leq i \leq \ell \text{ s.t.\ } \varphi \text{ is a pattern embedding}\}.
\end{equation}
\noindent
Finding pattern embeddings in port graphs is a simple problem
already studied in other contexts~\cite{Jiang1,Jiang2}.
As a result of \cref{eq:emb}, for every vertex $r$ in $P$ and $r_G$ in $G$ there can be at most one embedding $P \to G$
that maps $r$ to $r_G$:
for $v \neq r$ in $P$, there is a path in $P$ from $r$ to $v$ which, viewed as a sequence of port labels,
maps uniquely to a path in $G$ starting at $r_G$ and ending in the image of $v$.
For a choice of $r$ in $P$ it is therefore sufficient to consider every posible image $r_G$ in $G$ to find all
embeddings $P \to G$.

We are however interested in the regime where the number of patterns $\ell$ may be large and
pattern matching is performed many times for the same set of patterns.
For this scenario it makes sense to proceed in two steps 
and introduce
\emph{pattern matching with pre-computation}.
Given patterns
$P_1, \dots P_\ell$, we first produce a \emph{pattern matcher}, a program whose
representation can be stored to disk.
In a second step, a subject port graph $G$
is passed
to the pattern matcher, which
computes the set (\ref{eq:patternmatches}).
We are interested in two properties of the solution:
\begin{itemize}
    \item What is the complexity of pattern matching on input $G$ given such a pattern matcher?
    \item What is the time complexity of generating a pattern matcher given patterns,
    and what is the size of the pattern matcher that is produced?
\end{itemize}
The answer to the first question is our main concern: for a fixed collection of patterns, 
this will determine the runtime to obtain all pattern embeddings given an input
diagram. The second question, on the other hand, primarily concerns a one-off
pre-computation step that only needs to be performed once for any set of
patterns.
In practice, this may also impinge on the first question, if the matcher
does not fit into RAM and/or CPU cache.



\section{Algorithm description}\label{sec:toy}
\subsection{Canonical Tree Representation}\label{sec:ctrepr}
Connected port graphs admit an equivalent representation as trees, which we will
use for matching.


\paragraph{Split Graphs.}
Let $G$ be a connected port graph with vertices $V$ and
consider the linear path decomposition of $G$. In this decomposition every vertex $v$ in $G$
must be on one or more linear paths.
We mark a subset  $X \subseteq V$  of vertices of $G$ as `immutable' and split
every other vertex $v \in V \setminus X$
into multiple vertices, rewiring the edges in such a way that all vertices not in $X$ are now on
exactly one linear path.
We call the graph thus obtained the $X$-split graph of $G$, and write it $\textit{split}_X(G)$.
\Cref{fig:splitgraph} shows an example of a graph and its split graph.
\begin{figure}
\ctikzfig{split-graph} 
\caption{A port graph and its linear path decomposition (coloured edges) on the %
left. On the right, the split graph resulting from the choice of anchors $X = \{A, D\}$. %
We use the circuits convention, i.e. port labels are partitioned into linear paths %
using the relation $i_k \sim o_k$.}  
\label{fig:splitgraph}
\end{figure}
Formally, the split graph can be characterised using an equivalence relation $\equiv$
given by
\begin{equation}\label{eq:equiv-split}
  (v, p) \equiv (v', p')\quad\Leftrightarrow\quad v = v' \wedge (v \in X \vee p \sim p')
\end{equation}
The vertices of the split graph are the equivalence classes of $\equiv$
and the edges are obtained by mapping the edges of $G$ one to one:
two vertices in the $X$-split graph are connected by an edge if and only if
there is an edge in $G$ between some elements of their equivalence classes.
Note that if the anchor set $X$ is too small, $\textit{split}_X(G)$ may not be connected;
e.g. if $X = \varnothing$ there would be $\textit{width}(G)$  connected components, one per linear path.

A recovery of the original port graph $G$ given split$_X(G)$ 
is made possible by adding vertex labels to the split graph that identify split vertices.
In the following, split graphs are always rooted trees,
i.e. connected acyclic graphs with a chosen root vertex.
Using paths of port labels we obtain a vertex labelling that is 
invariant under pattern embedding.
This is discussed in more details in the context of the CT representation
in \cref{app:ctrepr}.



\begin{figure}
\begin{lstlisting}[%
    morekeywords={function, for, if, not, while, or, in, return, null},
    % morecomment=[l]{//},
    caption={Finding the set of canonical anchors. %
    \textsc{CanonicalAnchors} is a convenience wrapper around \textsc{ConsumePath}, which is defined recursively. %
    The latter returns not only the anchor list, but also the updated set of seen linear paths.
    \texttt{G.linear\_paths(v)} is %
    the set of all linear paths in $G$ that go through vertex $v$. For traversal,
    a linear path \texttt{lp} is split into two paths starting at vertex $v$ using
    \texttt{lp.split\_at(v)}. The sequence of vertices starting from $v$ to the end of the
    path is represented as a queue. %
    The symbol \texttt{++} designates list concatenation.},
    label=lst:anchors
]
function $\textsc{CanonicalAnchors}$(G: $\textit{Graph}$, root: $\textit{Vertex}$) -> $\textit{List[Vertex]}$:
  # Initialise the variables for ConsumePath and return the anchors
  (anchors, seen_paths) = $\textsc{ConsumePath}$(G, [root], $\varnothing$):
  return anchors

function $\textsc{ConsumePath}$(
    G: $\textit{Graph}$,
    path: $\textit{Queue[Vertex]}$,
    seen_paths: $\textit{Set[LinearPath]}$,
) $\to$ $\textit{(List[Vertex], Set[LinearPath])}$:
  new_anchor = null
  unseen = $\varnothing$
  # Find the first vertex in the queue on an unseen linear path
  while unseen == $\varnothing$:
    if path.is_empty():
      return ([], {})
    new_anchor = path.pop()
    unseen = G.linear_paths(new_anchor) $\setminus$ seen_paths

  # Add the new linear paths to the set of seen paths
  seen_paths = seen_paths $\cup$ unseen

  # We traverse the rest of current path as well as all the new linear paths
  paths = [path]
  for lp in unseen:
    (left_path, right_path) = lp.split_at(new_anchor)
    paths.push(left_path)
    paths.push(right_path)

  # For each path find anchors recursively and update seen paths
  anchors = [new_anchor]
  for path in paths:
    (new_anchors, new_seen_paths) = $\textsc{ConsumePath}$(G, path, seen_paths)
    anchors = anchors ++ new_anchors
    seen_paths = new_seen_paths
  return (anchors, seen_paths)
\end{lstlisting}
\end{figure}

\paragraph{Anchor sets.} If $\textit{split}_X(G)$ is connected and acyclic,
we say that $X$ is an anchor
set of $G$ and call the vertices in $X$ anchor vertices.
If $width(G) > 1$, then
every linear path in $G$ must contain at least one anchor
vertex.
A set of $width(G)$ anchors always exists and can be computed constructively:
\begin{restatable}{prop}{cananchors}\label{prop:cananchors}
For a connected port graph $G$ of width $w$ and depth $d$ and
for a vertex $r$ of $G$,
\cref{lst:anchors} returns an anchor set of $w$ vertices;
we call this the set of \textit{canonical anchors}.
Its runtime is $O(w^2 \cdot d)$.
\end{restatable}
\noindent
The proof is in \cref{app:cananchors}.
Note that the code given assumes that the linear paths of $G$ are already
computed. These can be computed at the beginning of the computation
in time linear in the graph size---and thus
do not affect the overall asymptotic complexity.
As a direct consequence we have the following:
\vspace{-4mm}
\begin{coro}\label{coro:width-anchors}
For a port graph $G$ and root $r_G$ in $G$:
$width(G) = |\textsc{CanonicalAnchors}(G, r_G)|$.\qed
\end{coro}

\paragraph*{CT representation.}
We call the tree split$_X(G)$ obtained from the canonical anchors, along with the choice of root $r$,
the \textbf{canonical tree} (CT) representation of $G$.
For simplicity, we will assume that the root is chosen such that it is on a single linear path\footnote{We can ensure that such a root always exists for example by adding dummy vertices on every edge.}.
Non-root internal nodes with more than one child are contained in the set of anchor vertices of $G$,
every leaf is at most a height $depth(G)$ below
the nearest anchor and there is at least one leaf
a distance $depth(G)$ from the nearest anchor.
As a consequence the CT tree width $t_w$ and height $t_h$ can be bound by 
\begin{equation}\label{eq:ctsize}
t_w \leq (\Delta - 2) \cdot width(G) = 2\cdot width(G)\quad\textrm{and}\quad depth(G) / 2 \leq t_h \leq width(G) \cdot depth(G),
\end{equation}
where $\Delta = 4$ is the maximum degree of $G$.
Using CT representations of patterns simplifies the pattern matching problem for three reasons:
\begin{itemize}
    \item Connected port graphs are uniquely characterised by their CT representation, i.e.
    there is an injective map from patterns with a choice of root to their CT representation.
    \item Every vertex in the CT representation is either the root vertex or it is uniquely identified by the path to it from the root. Paths can be defined by port labels,
    which are invariant under pattern embeddings.
    \item Rooted trees are uniquely characterised by a partition of their edges into paths, which can in turn be encoded as strings. See \cref{fig:ctstrings} for an example.
\end{itemize}
The properties of the CT representation is discussed in more details in \cref{app:ctrepr}.

\subsection{Pattern matching with fixed anchors}\label{sec:fixed-anchors}
We aim to present an $\ell$-independent pattern matcher: an algorithm that can identify all pattern embeddings in a subject graph~\eqref{eq:patternmatches} whose complexity is independent from the number of patterns $\ell$.
In this section, we restrict to pattern embeddings that map the set of canonical anchors of the pattern to a fixed subset of vertices in the subject graph, and we show that this problem can be reduced to a simple matching problem on strings.

We start by observing that such pattern embeddings with fixed anchors
correspond to tree inclusions of CT representations.
Let $G$ be a port graph, let $P_1, \dots, P_\ell$ be patterns of width $w$
and let $X \subseteq V$ be a set of $w$ vertices in $G$.
Choose root vertices $r_G \in X$
and $r_i$ in patterns $P_i$.
Write $T_i$ for the CT representation of $P_i$ rooted in $r_i$.
Consider the following set $\mathcal{G}$ of subgraphs of $G$:\footnote{A convex subgraph $H \subseteq G$ is one such the canonical embedding $H \to G$ is convex, as defined in assumption 3 of \cref{sec:assumptions}.}
\begin{equation}\label{eq:treeinc}
\mathcal{G} = \{ H \subseteq G\ | \ H \textrm{ is connected convex subgraph and }
\textsc{CanonicalAnchors}(H, r_G) = X\}.
\end{equation}

\begin{prop}\label{prop:treeinc}
If $\mathcal{G} \neq \varnothing$, then there is a connected subgraph $G_\text{max} \subseteq G$
such that $H \subseteq G_\text{max}$ for all $H \in \mathcal{G}$.
The split graph $\textit{split}_X(G_\text{max})$
is a tree rooted in $r_G$.
There is a pattern embedding $\varphi \colon P_i \to G$ mapping the canonical anchor set of $P_i$ to $X$ and $\varphi(r_i) = r_G$ if and only if there is an injective embedding of trees $\phi \colon T_i \to \textrm{split}_X(G_\text{max})$
with $\phi(r_i) = r_G$ that satisfies \cref{eq:emb} and  preserves vertex labels.
\end{prop}
\noindent
The proof gives an explicit construction for $G_\text{max}$.
\begin{proof}
Let $L_1, \dots, L_w$ be the subset of linear paths in $G$ that go through at least one vertex in $X$.
Let $G_{\text{max}} \subseteq G$ be the subgraph of $G$ defined by the edges
\begin{equation}\label{eq:Lsubgraph}
E_\text{max} = \bigcup_{1 \leq i \leq w} L_i,
\end{equation}
For any subgraph $H \in \mathcal{G}$ it holds that $H \subseteq G_\text{max}$ 
due to the linear paths of $H$ being contained in $L_1, \dots, L_w$.
Since any $H \in \mathcal{G}$ is connected, the anchors in $X$ are connected in $H$
and therefore also in $G_{\text{max}}$. As a consequence, all vertices of $G_\text{max}$ are connected.
The port graph split$_X(G_{\text{max}})$ must be a tree, as otherwise its canonical anchors
would be a strict subset of $X$ and by \cref{coro:width-anchors},
$width(\text{split}_X(G_\text{max})) < |X|$. Hence,
\[
width(G_\text{max}) = width(\text{split}_X(G_\text{max})) < |X| = width(H).
\]
contradicting the convexity assumption of \cref{eq:treeinc}.
Assuming $\mathcal{G} \neq \varnothing$,
we now prove the bidirectional implication between $\varphi$ and $\phi$.

$\Leftarrow$:
We use vertex labels on $T_i$ and split$_X(G_\text{max})$ to mark with a unique label
vertices that were split from a same vertex $v$ in $P_i$, respectively $G_\text{max}$
(details in \cref{app:ctrepr}).
Let $\mathcal{V}_{P_i}$ and $\mathcal{V}_{G_\text{max}}$
be the partition of the vertices of $T_i$ and split$_X(G_\text{max})$
into sets of split vertices with identical labels;
there are bijective maps between $\mathcal{V}_{P_i}$ and
the vertices in $P_i$ as well as between $\mathcal{V}_{G_\text{max}}$ and the vertices in $G_\text{max}.$
The tree embedding $\phi \colon T_i \to \text{split}_X(G_\text{max})$ preserves vertex labels
and thus maps sets in $\mathcal{V}_{P_i}$ to sets
in $\mathcal{V}_{G_\text{max}}$: it thus defines
a map $\varphi_V: P_i \to G_\text{max}$.

$\varphi_V$ is injective by injectivity of $\phi$ and maps the root $r_i$ to $r_G$ by construction.
The $w-1$ non-root anchor vertices
of $T_i$ are the
only vertices in $T_i$ on more than
one linear path: they must be mapped
to the $w-1$ vertices in split$_X(G_\text{max})$ with the same properties---precisely its non-root anchor vertices.
Edges are mapped bijectively by graph splitting and thus
the map $\varphi_E$ can be defined by
using $\phi_E$ and must satisfy
\cref{eq:emb}.
Since $G_\text{max} \subseteq G$, we conclude that $\varphi$ is a valid pattern embedding
$P_i \to G$.

$\Rightarrow$:
The image of $\varphi: P_i \to G$ is a convex connected subgraph of $G$ with canonical anchors
$X$ and root $r_G$, and thus is in $\mathcal{G}$.
It must in particular be a subgraph of $G_\text{max}$, and thus we can view $\varphi$
as an embedding $\varphi: P_i \to G_\text{max}$.

Note that edges are mapped bijectively between split and unsplit graphs,
and thus the pattern embedding $\varphi$ defines an injective map $\phi_E$
from edges in $T_i$ to edges in split$_X(G_\text{max})$.
We construct the map $\phi: T_i \to \text{split}_X(G_\text{max})$ inductively over the vertex set of $T_i$.
We start by defining $\phi(r) := r_G$.
Using \cref{eq:emb} and $\phi_E$,
we can then uniquely define the image of any neighbouring vertex of $r$ in $T_i$.
Proceeding inductively we will define $\phi$ on all vertices of $T_i$ since it is connected.
Because \cref{eq:emb} holds on $\varphi$, this procedure is well-defined and the resulting
map $\phi$ will also satisfy \cref{eq:emb}.

Now suppose $v, v'$ are vertices in $T_i$ such that $\phi(v) = \phi(v')$.
By the inductive construction there are paths from the root $r$ to $v$ and $v'$ respectively
such that their image under $\phi_E$ are two paths from $r_G$ to $\phi(v) = \phi(v')$.
But split$_X(G_\text{max})$ is a tree, so both paths must be equal. By bijectivity of $\phi_E$,
it follows $v = v'$, and thus $\phi$ is injective.
Finally, the vertex labels are defined to be invariant under pattern embedding and thus are
preserved by definition.\qedhere
\end{proof}
\noindent
Given $G$ and a vertex set $X$ we can thus find a maximal subgraph $G_\text{max} \subseteq G$
that contains all subgraphs of $G$ with canonical anchors $X$.
Given $P_1, \dots, P_\ell$, $X$ and $G$, we then
compute the CT representations of $P_1, \dots, P_\ell$
and check for inclusions within $\textrm{split}_X(G_\text{max})$.
We will use an $\ell$-independent tree matching algorithm for the latter task, thus solving the $\ell$-independent pattern
matching problem on port graphs.

\paragraph{String encoding of CT representations.}
We reduce the tree inclusion problem that results from \cref{prop:treeinc}
to a string prefix matching problem that admits a well-known solution, discussed in 
\cref{prop:prefixmatch} of \cref{app:prefixmatch}.
The main idea is to partition the edges of the CT representation into linear paths,
each of which is represented
by two strings.
They encode the paths from the anchor vertex on the path to either end of the linear path
by expressing them as sequences of port labels.
A graph of width $w$ will have $w$ linear paths and will be split into $2w$ strings.
For the example graph of \cref{fig:splitgraph}, we obtain six split linear paths,
shown in \cref{fig:ctstrings}.
See \cref{app:ctrepr} for more details.

\begin{figure}
    \centering
    \resizebox{0.7\textwidth}{!}{\begin{tikzpicture}
	\begin{pgfonlayer}{nodelayer}
		\node [style=none] (76) at (-11.75, 0.5) {};
		\node [style=none] (77) at (-10.375, 0.5) {};
		\node [style=deg2-gate-grey] (78) at (-9.75, 1) {\footnotesize $A$};
		\node [style=none] (98) at (-10.5, 0.475) {$i_2$};
		\node [style=none] (111) at (-11.75, 4) {};
		\node [style=none] (112) at (-10.375, 4) {};
		\node [style=deg2-gate-grey] (113) at (-9.75, 3.5) {\footnotesize $A$};
		\node [style=none] (114) at (-10.5, 3.975) {$i_1$};
		\node [style=deg2-gate-grey] (117) at (-5.75, 3.5) {\footnotesize $A$};
		\node [style=none] (118) at (-4.95, 3.975) {$o_1$};
		\node [style=none] (124) at (-3.675, 3.975) {};
		\node [style=deg1-gate-half] (125) at (-3, 4) {\ \ \footnotesize $B_1$};
		\node [style=none] (127) at (-3.4, 3.975) {$i_1$};
		\node [style=deg2-gate-grey] (129) at (-9.775, -1.425) {\footnotesize $D$};
		\node [style=none] (130) at (-10.5, -1.925) {$i_2$};
		\node [style=none] (133) at (-12.425, -1.925) {};
		\node [style=deg1-gate-half] (134) at (-12.925, -1.9) {{\footnotesize $B_2$}\ \ \ \ };
		\node [style=none] (135) at (-12.425, -1.95) {$o_2$};
		\node [style=deg2-gate-grey] (136) at (-5.75, -1.425) {\footnotesize $D$};
		\node [style=none] (137) at (-4.95, -1.95) {$o_2$};
		\node [style=none] (138) at (-3.65, -1.95) {};
		\node [style=deg2-gate-grey] (139) at (-5.75, 1) {\footnotesize $A$};
		\node [style=none] (140) at (-3.8, 0.475) {};
		\node [style=none] (141) at (-4.95, 0.475) {$o_2$};
		\node [style=deg1-gate] (142) at (0.375, 0.425) {\footnotesize $C_1$};
		\node [style=none] (143) at (-0.8, 0.425) {};
		\node [style=none] (145) at (-1.6, 0.425) {};
		\node [style=none] (147) at (1.175, 0.375) {$o_1$};
		\node [style=none] (148) at (4, 1.55) {};
		\node [style=deg2-gate-grey] (149) at (4.5, 1.05) {\footnotesize $D$};
		\node [style=none] (150) at (3.775, 1.55) {$i_1$};
		\node [style=none] (151) at (5.3, 1.525) {$o_1$};
		\node [style=none] (152) at (6.6, 1.025) {};
		\node [style=none] (153) at (-0.35, 0.425) {$i_1$};
		\node [style=deg1-gate] (154) at (-2.625, 0.475) {\footnotesize $C_2$};
		\node [style=none] (155) at (-1.8, 0.425) {$o_2$};
		\node [style=none] (156) at (-3.35, 0.475) {$i_2$};
		\node [style=none] (158) at (-9.75, 5.25) {\itshape left half};
		\node [style=none] (159) at (-5.6, 5.25) {\itshape right half};
		\node [style=none] (160) at (-7.75, 5.75) {};
		\node [style=none] (161) at (-7.75, -3.1) {};
		\node [style=none] (162) at (-17.15, 3.5) {linear path 1:};
		\node [style=none] (163) at (-17.15, 1) {linear path 2:};
		\node [style=none] (164) at (-17.15, -1.425) {linear path 3:};
	\end{pgfonlayer}
	\begin{pgfonlayer}{edgelayer}
		\draw [style=edge2] (76.center) to (77.center);
		\draw [style=edge1] (111.center) to (112.center);
		\draw [style=edge1, in=180, out=0] (118.center) to (124.center);
		\draw [style=edge3, in=180, out=0, looseness=1.50] (133.center) to (130.center);
		\draw [style=edge3] (137.center) to (138.center);
		\draw [style=edge2, in=0, out=180, looseness=1.25] (140.center) to (141.center);
		\draw [style=edge2, in=180, out=0, looseness=1.75] (147.center) to (150.center);
		\draw [style=edge2] (151.center) to (152.center);
		\draw [style=edge2] (143.center) to (155.center);
		\draw (161.center) to (160.center);
	\end{pgfonlayer}
\end{tikzpicture}}
    \caption{The split graph of \cref{fig:splitgraph}, represented by 
    6 strings obtained from its 3 linear paths.}
    \label{fig:ctstrings}
\end{figure}


This encoding defines the \textsc{AsStrings}$(T, r)$ procedure:
it takes as input a connected acyclic split graph $T$ and a root $r$ in $T$,
and returns an encoding of $T$ and its vertex labels as $2 \cdot width(T)$ strings.
\begin{restatable}{prop}{ctstrings}\label{prop:ctstrings}
    Let $T_1, T_2$ be acyclic connected split graphs of width $w$
    and let $r_1, r_2$ be vertices in $T_1$ resp. $T_2$.
    Let $(s_1, \dots, s_{2w}) = \textsc{AsStrings}(T_1, r_1)$ and $(t_1, \dots, t_{2w}) = \textsc{AsStrings}(T_2, r_2)$
    be the string encodings of their linear paths.
    Then there is an injective tree embedding $T_1 \to T_2$
    that satisfies \cref{eq:emb}, maps $r_1$ to $r_2$ and preserves vertex labels
    if and only if $s_i \subseteq t_i$ for all $1 \leq i \leq 2w$.
\end{restatable}
\noindent
The proof consists in showing that trees can be fully defined by the set of all paths
from the root, which can be encoded in and recovered from the string representation.
The proof is in \cref{app:ctstrings}.

The $\subseteq$ notations on string designates prefix inclusion.
The string prefix matching problem is a simple computational task that can be generalised
to to check for multiple string patterns at the same time.
An overview of this problem can be found in \cref{app:prefixmatch}.
Putting \cref{prop:treeinc,prop:ctstrings} together, we can thus obtain a solution
for the $\ell$-independent pattern matching problem for fixed anchors:
\begin{prop}\label{prop:fixedanchors}
Let $G$ be a port graph, $P_1, \dots, P_\ell$ be patterns of width $w$ and depth at most $d$, and
$X \subseteq V$ be a set of $w$ vertices in $G$.
The set of all pattern embeddings mapping the canonical anchor set of $P_i$ to $X$
and root $r_i$ to $r_G$ for $1 \leq i \leq \ell$
can be computed in time $O(w\cdot d)$ using a pre-computed prefix tree of size
at most $(\ell \cdot d + 1)^w$, 
constructed in time complexity $O((\ell \cdot d)^w)$.\qed
\end{prop}

\subsection{Enumeration of anchor sets}\label{sec:anchors}
Assume that all patterns have at most width $w$ and depth $d$.
All that remains to turn \cref{prop:fixedanchors} into a complete
solution for pattern matching is to enumerate all possible sets $X$ of
at most $w$ vertices in $G$ that are the canonical anchors of some
subgraph of $G$ of width $w$.
The bound on the number of such sets (\cref{prop:catalanbound}) is one of the key stepping stones of this paper
that makes $\ell$-independent matching possible on port graphs.

\paragraph{Procedure.}
We introduce \textsc{AllAnchors}, a procedure similar to $\textsc{CanonicalAnchors}$ of \cref{lst:anchors},
described in \cref{lst:extract} in detail.
\textsc{AllAnchors} takes as input a port graph $G$, a root vertex $r_G$ and a width $w \geq 1$,
and returns all sets of $w$ vertices that form the canonical anchors of some subgraph of $G$
with CT representation rooted at $r_G$.
The main difference between \cref{lst:anchors,lst:extract} is that the successive
calls to \textsc{ConsumePath} on line 35 of \cref{lst:anchors} 
are replaced by a series of nested loops (lines 42--48 in \cref{lst:extract})
that exhaustively iterate over the possible outcomes for different subgraphs of $G$.
The results of every possible combination of recursive calls are then collected
into a list of anchor sets, which is returned.

\begin{figure}
\begin{lstlisting}[%
    morekeywords={function, for, if, not, assert, in, return, such, that, while, len, else},
    caption={%
      Returns all sets of $w$ vertices that form the canonical anchors of some subgraph of $G$
      with CT representation rooted at $r_G$.. The code structure mirrors \cref{lst:anchors},
      with
      \textsc{AllAnchors} and \textsc{AllConsumePath} replacing %
      \textsc{CanonicalAnchors} and \textsc{ConsumePath} respectively. %
      \texttt{lp.split\_at}, \texttt{G.linear\_paths} and \texttt{++} are defined as in \cref{lst:anchors}.},
    label=lst:extract
]
function $\textsc{AllAnchors}$(G:$\,\textit{Graph}$, root:$\,\textit{Vertex}$, w:$\,\textit{Integer}$) -> $\textit{List[List[Vertex]]})$:
  # Assumption: root is on a single linear path
  assert len(G.linear_paths(root)) == 1

  # Initialise the variables for AllConsumePath and return the anchor lists
  all_anchors = []
  for (anchors, seen_paths) in $\textsc{AllConsumePath}$(G, w, [root], $\varnothing$):
    all_anchors.push(anchors)
  return all_anchors

function $\textsc{AllConsumePath}$(
    G: $\textit{Graph}$,
    w: $\textit{Integer}$,
    path: $\textit{Queue[Vertex]}$,
    seen_paths: $\textit{Set[LinearPath]}$,
) $\to$ $\textit{List[(List[Vertex], Set[LinearPath])]}$:
  # Base case: return one empty anchor list
  if w == 0:
    return [[]]

  new_anchor = null
  unseen = $\varnothing$
  # Find the first vertex in the queue on an unseen linear path
  while unseen == $\varnothing$:
    if path.is_empty():
      return []
    new_anchor = path.pop()
    unseen = G.linear_paths(new_anchor) $\setminus$ seen_paths
  # Every vertex is on at most one unseen linear path as either
  #  - the new anchor is the root, in which case it is on at most one linear path
  #  - or it is on up to two linear paths, but one of them has already been seen.
  assert len(unseen) == 1
  new_path = unseen[0]

  # The w anchors are made of the new anchor and w-1 anchors on path1 - path3
  path1 = path
  path2, path3 = new_path.split_at(new_anchor)
  seen0 = seen_paths $\cup$ {new_path}
  return_list = []
  # Iterate over all ways to split w-1 anchors over the three paths
  # and solve recursively
  for 0 $\leq$ w1, w2, w3 < w such that w1 + w2 + w3 == w - 1:
    for (anchors1, seen1) in $\textsc{AllConsumePath}$(G, w1, path1, seen0):
      for (anchors2, seen2) in $\textsc{AllConsumePath}$(G, w2, path2, seen1):
        for (anchors3, seen3) in $\textsc{AllConsumePath}$(G, w3, path3, seen2):
          # Concatenate new anchor with anchors from all paths
          anchors = [new_anchor] ++ anchors1 ++ anchors2 ++ anchors3
          return_list.push(anchors, seen3)
  return return_list
\end{lstlisting}
\end{figure}

\begin{restatable}[Correctness of \textsc{AllAnchors}]{prop}{allanchors}\label{prop:allanchors}
    Let $G$ be a port graph and $H \subseteq G$ be a connected convex subgraph of $G$ of width $w$.
    Let $r$ be a vertex of $H$. We have
    $\textsc{CanonicalAnchors}(H, r) \in \textsc{AllAnchors}(G, r, w).$
\end{restatable}
\noindent
The proof is by induction over the width $w$ of the subgraph $H$ and given in \cref{app:allanchors}.
The idea is to map every recursive call to \textsc{ConsumePath} in \cref{lst:anchors}
to a call to \textsc{AllConsumePath} in lines 42--48 of \cref{lst:extract}.
All recursive results are concatenated on line 47, and thus the value returned by \textsc{ConsumePath}
will be one of the anchor sets in the list returned by \textsc{AllConsumePath}.

We will see that the overall runtime complexity of \textsc{AllAnchors} can be easily
derived from a bound on the size of the returned list. 
For this we use the following result:
\begin{prop}\label{prop:catalanbound}
    For a port graph $G$ and vertex $r_G$ in $G$,
    the length of the list $\textsc{AllAnchors}(G, r_G, w)$ is in
    $O(c^w \cdot w^{-\sfrac32})$, where $c = 6.75$ is a constant.
\end{prop}
    
\begin{proof}
Let $C_w$ be an upper bound for the length of the list returned by
a call to \textsc{AllConsumePath} for width $w$,
and thus a bound on the length of the list returned by \textsc{AllAnchors}.
For the base case $w = 0$, $C_0 = 1$.
The returned \texttt{all\_anchors} list is obtained by pushing anchor lists
one by one on line 48.
We can count the number of times this line is executed by multiplying
the length of the lists returned by the recursive calls on lines 43--45, giving
us the recursion relation
\begin{equation}\label{eq:catalanrec}
  C_w \leq \sum_{\substack{0 \leq w_1, w_2, w_3 < w\\w_1 + w_2 + w_3 = w - 1}} C_{w_1} \cdot C_{w_2} \cdot C_{w_3}.
\end{equation}
Since $C_w$ is meant to be an upper bound, we replace $\leq$ with equality in \cref{eq:catalanrec} to obtain a recurrence relation for $C_w$.
This recurrence relation is a generalisation of the well-known
Catalan numbers~\cite{catalan}, equivalent to counting the number of ternary trees with $w$ internal nodes:
a ternary tree with $w \geq 1$ internal nodes is made of a root along with
three subtrees with $w_1,w_2$ and $w_3$ internal nodes respectively, with
$w_1 + w_2 + w_3 = w-1$.
A closed form solution to this problem can be found in~\cite{fusscatalan}:
\[
C_w = \frac{{3w \choose w}}{2w + 1} = \Theta \left(\frac{c^w}{w^{\sfrac32}} \right)
\]
satisfies \cref{eq:catalanrec} with equality,
where $c = \sfrac{27}{4} = 6.75$ is a constant obtained from the Stirling approximation:
\[
{3w \choose w} = \frac{(3w)!}{(2w)!w!} = \Theta\left(\frac{1}{\sqrt{w}}\right)
\Big(\frac{(3w)^3}{e^3}\Big)^{w}\Big(\frac{e^2}{(2w)^2}\Big)^{w}\Big(\frac{e}{w}\Big)^{w}
= \Theta\left(\frac{(\sfrac{27}{4})^w}{w^{\sfrac12}}\right).\qedhere
\]
\end{proof}
\noindent
To obtain a runtime bound for \textsc{AllAnchors}, it is useful to identify how much of $G$ needs to be traversed.
If we suppose all patterns have at most depth $d$, then it immediately follows that any
vertex in $G$ that is in the image of a pattern embedding must be at most a distance $d$ away
from an anchor vertex.
For this purpose, we modify the definition of \texttt{split\_at} in \cref{lst:extract} to only return the first $d$ vertices of any path returned.
We thus obtain the following runtime.
\begin{restatable}{coro}{allanchorscoro}\label{prop:allanchorscoro}
For patterns with at most width $w$ and depth $d$,
the total runtime of \textsc{AllAnchors} is in
\begin{equation}\label{eq:catalan}
    O\left(\frac{c^w \cdot d}{w^{\sfrac12}}\right).
\end{equation}
\end{restatable}
\noindent
The proof is in \cref{app:allanchorscoro}.
Finally, we reach our main result.
\mainthm
\begin{proof}
The pre-computation consists of running the \textsc{CanonicalAnchors} procedure on
every pattern and then transforming them into string tuples using \textsc{AsStrings}.
\textsc{AsStrings} is linear in pattern sizes and \textsc{CanonicalAnchors} runs in $O(w^2\cdot d)$ for each pattern (\cref{prop:cananchors}).
This is followed by the insertion of $\ell$ tuples of $2w$ strings of length $\Theta(d)$
into a multidimensional prefix tree. This dominates the total runtime, which can be obtained
directly from \cref{prop:prefixmatch}.

The complexity of pattern matching itself on the other hand is composed of two parts:
the computation of all anchor set candidates, and the execution of
the prefix string matcher for each of the trees resulting from these sets of fixed anchors.
The complexity of the former is obtained by
multiplying the result of \cref{prop:catalanbound} with $|G|$,
as \textsc{AllAnchors}
must be run for every choice of root vertex $r$ in $G$:
\begin{equation}\label{eq:finalcomplexity}
  O(w \cdot d \cdot C_w \cdot |G|),
\end{equation}
where $C_w$ is the bound for the number of anchor lists returned by \textsc{AllAnchors}.
For the latter we use \cref{prop:prefixmatch} and obtain the complexity
$O(w \cdot d \cdot C_w)$, which is dominated by \cref{eq:finalcomplexity}.
\end{proof}

\section{Pattern matching in practice}\label{sec:impl}
\begin{figure}[t]
    \centering
    \includegraphics[width=\textwidth]{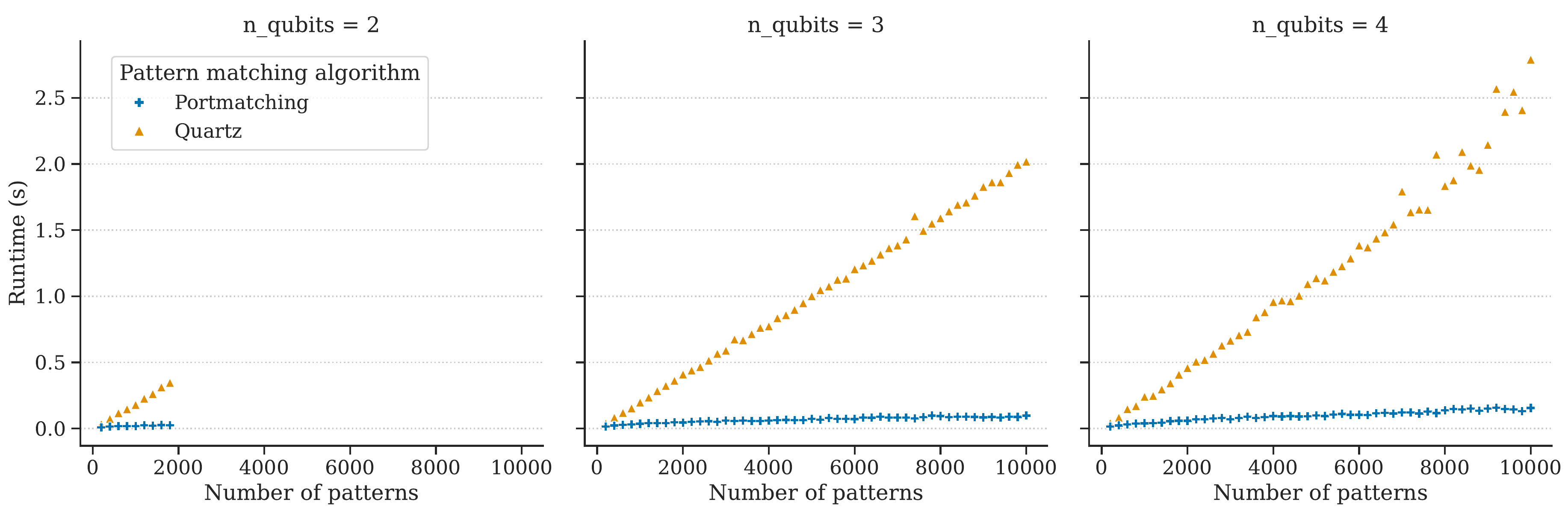}
    \caption{Runtime of pattern matching for $\ell = 0\dots 10^4$
    patterns on 2, 3 and 4 qubit quantum circuits from the Quartz ECC dataset, %
    for our implementation (Portmatching) and the Quartz project. %
    All $\ell = 1954$ two qubit circuits were used, whereas for 3 and 4 qubit circuits,
    $\ell = 10^4$ random samples were drawn.}
    \label{fig:benchmarks}
\end{figure}
\Cref{thm:main} shows that pattern independent matching can scale to large
datasets of patterns but imposes some restrictions on the patterns and embeddings
that can be matched.
In this section we discuss these limitations and
give empirical evidence that the pattern matching approach we have presented
can be used on a large scale, outperforming existing solutions.

\paragraph{Pattern limitations.}
In \cref{sec:assumptions}, we imposed conditions on the pattern embeddings
in order to obtain a complexity bound for pattern independent matching.
We argued how these restrictions are natural for applications in quantum computing and most
of the arguments will also hold for a much broader class of computation graphs.

In future work, it would nonetheless be of theoretical interest to explore the importance
of these assumptions and their impact on the complexity of the problem.
As a first step towards a generalisation, our implementation
and all our benchmarks in this section do not make any of these simplifying assumptions.
Our results below give empirical evidence that 
a significant performance advantage can be obtained regardless.

\paragraph{Implementation.}
We provide an open source implementation in Rust of pattern independent matching
using the results of \cref{sec:toy}, described in more detail in \cref{app:portmatching}.
The implementation works for weighted or unweighted port graphs,
and makes none of the simplifying assumptions employed in the theoretical analysis.

\paragraph{Benchmarks.}
To assess practical use, we have benchmarked our implementation against a
leading C++ implementation of pattern matching for quantum circuits
from the Quartz superoptimiser project~\cite{quartz}.
Using a real world dataset of patterns obtained by the Quartz equivalence classes
of circuits (ECC) generator, we measured the pattern matching runtime on a
random subset of up to \num{10000} patterns.
We considered circuits on the $T, H, CX$ gate set with up to 6 gates and 2, 3 or 4 qubits.
Thus for our patterns we have the bound $d \leq 6$ for the maximum depth and width $w = 2,3,4$.
In all experiments the graph $G$ subject to pattern matching was \texttt{barenco\_tof\_10} input, i.e.
a 19 qubit circuit input with 674 gates obtained by decomposing a 10-qubit Toffoli gate using the
Barenco decomposition~\cite{barenco}.
The results are summarised in \cref{fig:benchmarks}.
For $\ell = 200$ patterns, our proposed algorithm is $3\times$ faster than Quartz,
scaling up to $20\times$ faster for $\ell=10^5$.

We also provide a more detailed scaling analysis of our implementation
by generating random sets of \num{10000} quantum circuits with 15 gates
for qubit numbers between $w=2$ and $w=10$, using the previous gate set;
the results are shown in \cref{fig:scaling}.
From \cref{thm:main}, we expect that the pattern matching runtime is
upper bounded by a $\ell$-independent constant.
Runtime seems indeed to saturate for $w=2$ and $w=3$ qubit patterns, 
with an observable runtime plateau at large $\ell$.
From the exponential $c^w$ dependency in \cref{eq:mainruncompl}, it is however to be expected
that this upper bound increases rapidly for qubit counts $w \geq 4$.
A runtime ceiling is not directly observable at this experiment size but 
the gradual decrease in the slope
of the curve is consistent with the existence of the $\ell$-independent upper bound predicted
in \cref{thm:main}.

\begin{figure}[t]
    \centering
    \includegraphics[width=0.55\textwidth]{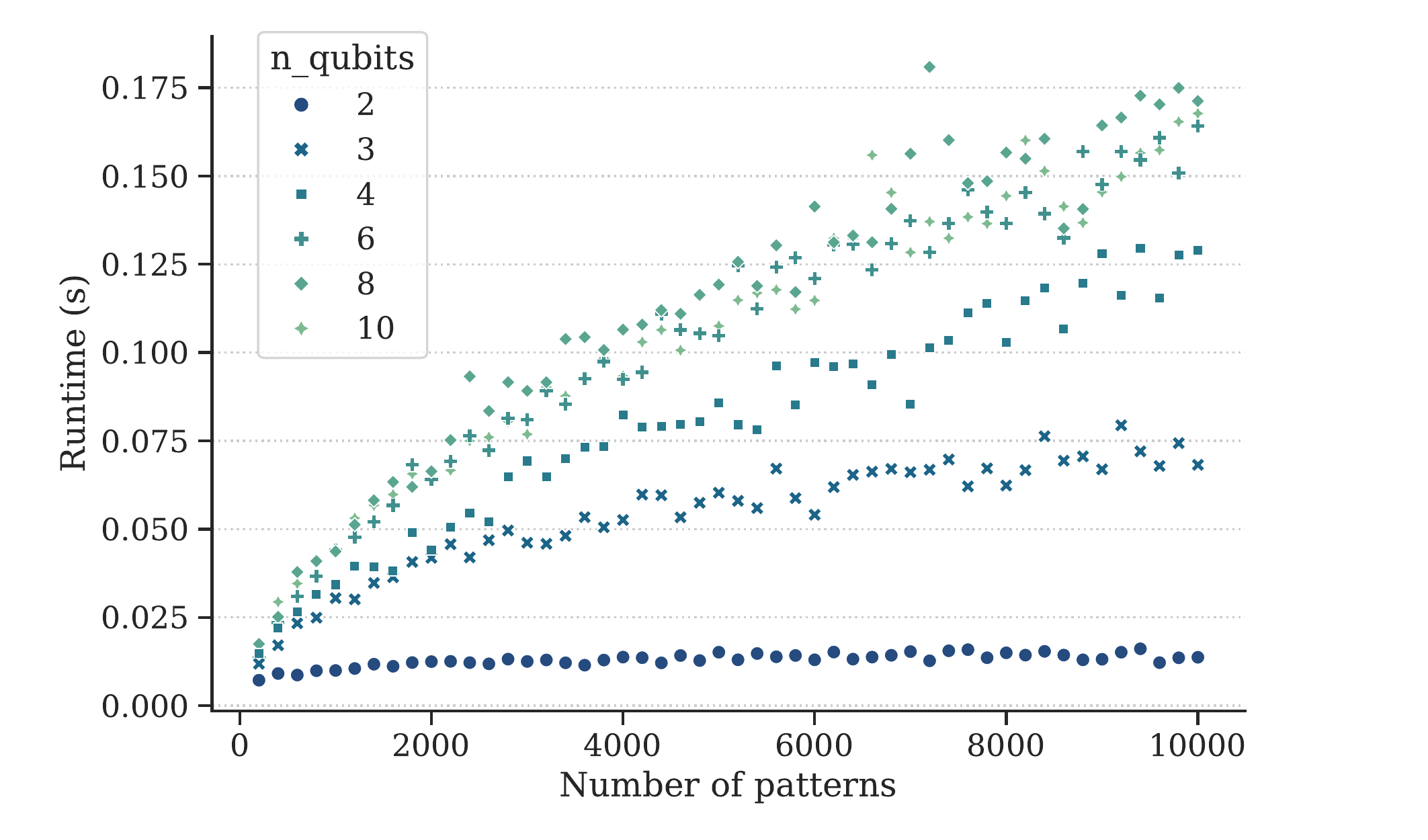}
    \caption{Runtime of our pattern matching for random quantum circuits with up to 10 qubits.}
    \label{fig:scaling}
\end{figure}

\section{Conclusion}
We have demonstrated that pattern matching on port graphs can be done in a runtime asymptotically independent of
the number of patterns by pre-computing an automaton-like data structure.
As a result, we obtain a provable computational advantage in the regime
of numerous low-width patterns.
This opens up promising avenues for graph rewriting and particularly for the optimisation of computation
graphs and quantum circuits.
Benchmarks further show that the approach is fast in practice.
At the scale of interest (\num{10000} pattern circuits with 3-4 qubits),
the resulting implementation of pattern matching on quantum circuits is 20x times faster than that of Quartz~\cite{quartz}, a leading quantum superoptimiser.

\bibliographystyle{eptcs}
\bibliography{refs}

\begin{thebibliography}{10}
\providecommand{\bibitemdeclare}[2]{}
\providecommand{\surnamestart}{}
\providecommand{\surnameend}{}
\providecommand{\urlprefix}{Available at }
\providecommand{\url}[1]{\texttt{#1}}
\providecommand{\href}[2]{\texttt{#2}}
\providecommand{\urlalt}[2]{\href{#1}{#2}}
\providecommand{\doi}[1]{doi:\urlalt{https://doi.org/#1}{#1}}
\providecommand{\eprint}[1]{arXiv:\urlalt{https://arxiv.org/abs/#1}{#1}}
\providecommand{\bibinfo}[2]{#2}

\bibitemdeclare{article}{fusscatalan}
\bibitem{fusscatalan}
\bibinfo{author}{Jean-Christophe \surnamestart Aval\surnameend}
  (\bibinfo{year}{2008}): \emph{\bibinfo{title}{Multivariate Fuss–Catalan
  numbers}}.
\newblock {\slshape \bibinfo{journal}{Discrete Mathematics}}
  \bibinfo{volume}{308}(\bibinfo{number}{20}), p. \bibinfo{pages}{4660–4669},
  \doi{10.1016/j.disc.2007.08.100}.

\bibitemdeclare{article}{barenco}
\bibitem{barenco}
\bibinfo{author}{Adriano \surnamestart Barenco\surnameend},
  \bibinfo{author}{Charles~H. \surnamestart Bennett\surnameend},
  \bibinfo{author}{Richard \surnamestart Cleve\surnameend},
  \bibinfo{author}{David~P. \surnamestart DiVincenzo\surnameend},
  \bibinfo{author}{Norman \surnamestart Margolus\surnameend},
  \bibinfo{author}{Peter \surnamestart Shor\surnameend}, \bibinfo{author}{Tycho
  \surnamestart Sleator\surnameend}, \bibinfo{author}{John~A. \surnamestart
  Smolin\surnameend} \& \bibinfo{author}{Harald \surnamestart
  Weinfurter\surnameend} (\bibinfo{year}{1995}):
  \emph{\bibinfo{title}{Elementary gates for quantum computation}}.
\newblock {\slshape \bibinfo{journal}{Phys. Rev. A}} \bibinfo{volume}{52}, pp.
  \bibinfo{pages}{3457--3467}, \doi{10.1103/PhysRevA.52.3457}.

\bibitemdeclare{article}{bonchiI}
\bibitem{bonchiI}
\bibinfo{author}{Filippo \surnamestart Bonchi\surnameend},
  \bibinfo{author}{Fabio \surnamestart Gadducci\surnameend},
  \bibinfo{author}{Aleks \surnamestart Kissinger\surnameend},
  \bibinfo{author}{Pawel \surnamestart Sobocinski\surnameend} \&
  \bibinfo{author}{Fabio \surnamestart Zanasi\surnameend}
  (\bibinfo{year}{2020}): \emph{\bibinfo{title}{String Diagram Rewrite Theory
  {I}: Rewriting with {F}robenius Structure}}.
\newblock {\slshape \bibinfo{journal}{Journal of the ACM (JACM)}}
  \bibinfo{volume}{69}, pp. \bibinfo{pages}{1 -- 58}, \doi{10.1145/3502719}.

\bibitemdeclare{article}{bonchiII}
\bibitem{bonchiII}
\bibinfo{author}{Filippo \surnamestart Bonchi\surnameend},
  \bibinfo{author}{Fabio \surnamestart Gadducci\surnameend},
  \bibinfo{author}{Aleks \surnamestart Kissinger\surnameend},
  \bibinfo{author}{Pawel \surnamestart Sobocinski\surnameend} \&
  \bibinfo{author}{Fabio \surnamestart Zanasi\surnameend}
  (\bibinfo{year}{2021}): \emph{\bibinfo{title}{String diagram rewrite theory
  {II}: Rewriting with symmetric monoidal structure}}.
\newblock {\slshape \bibinfo{journal}{Mathematical Structures in Computer
  Science}} \bibinfo{volume}{32}, pp. \bibinfo{pages}{511 -- 541},
  \doi{10.1017/S0960129522000317}.

\bibitemdeclare{inproceedings}{computationgraph}
\bibitem{computationgraph}
\bibinfo{author}{Jin \surnamestart Fang\surnameend}, \bibinfo{author}{Yanyan
  \surnamestart Shen\surnameend}, \bibinfo{author}{Yue \surnamestart
  Wang\surnameend} \& \bibinfo{author}{Lei \surnamestart Chen\surnameend}
  (\bibinfo{year}{2020}): \emph{\bibinfo{title}{Optimizing {DNN} computation
  graph using graph substitutions}}.
\newblock In: {\slshape \bibinfo{booktitle}{Proceedings of the VLDB
  Endowment}}, \bibinfo{volume}{13}, pp. \bibinfo{pages}{2734 -- 2746},
  \doi{10.14778/3407790.3407857}.

\bibitemdeclare{article}{portgraph}
\bibitem{portgraph}
\bibinfo{author}{Maribel \surnamestart Fernández\surnameend},
  \bibinfo{author}{Hélène \surnamestart Kirchner\surnameend} \&
  \bibinfo{author}{Bruno \surnamestart Pinaud\surnameend}
  (\bibinfo{year}{2018}): \emph{\bibinfo{title}{{Labelled Port Graph – A
  Formal Structure for Models and Computations}}}.
\newblock {\slshape \bibinfo{journal}{Electronic Notes in Theoretical Computer
  Science}} \bibinfo{volume}{338}, pp. \bibinfo{pages}{3--21},
  \doi{10.1016/j.entcs.2018.10.002}.
\newblock \bibinfo{note}{The 12th Workshop on Logical and Semantic Frameworks,
  with Applications (LSFA 2017)}.

\bibitemdeclare{article}{reteforgy}
\bibitem{reteforgy}
\bibinfo{author}{Charles~L. \surnamestart Forgy\surnameend}
  (\bibinfo{year}{1982}): \emph{\bibinfo{title}{Rete: A fast algorithm for the
  many pattern/many object pattern match problem}}.
\newblock {\slshape \bibinfo{journal}{Artificial Intelligence}}
  \bibinfo{volume}{19}(\bibinfo{number}{1}), pp. \bibinfo{pages}{17--37},
  \doi{10.1016/0004-3702(82)90020-0}.

\bibitemdeclare{inproceedings}{taso}
\bibitem{taso}
\bibinfo{author}{Zhihao \surnamestart Jia\surnameend}, \bibinfo{author}{Oded
  \surnamestart Padon\surnameend}, \bibinfo{author}{James~J. \surnamestart
  Thomas\surnameend}, \bibinfo{author}{Todd \surnamestart
  Warszawski\surnameend}, \bibinfo{author}{Matei~A. \surnamestart
  Zaharia\surnameend} \& \bibinfo{author}{Alexander \surnamestart
  Aiken\surnameend} (\bibinfo{year}{2019}): \emph{\bibinfo{title}{{TASO}:
  optimizing deep learning computation with automatic generation of graph
  substitutions}}.
\newblock In: {\slshape \bibinfo{booktitle}{Proceedings of the 27th ACM
  Symposium on Operating Systems Principles}}, \doi{10.1145/3341301.3359630}.

\bibitemdeclare{inproceedings}{Jiang1}
\bibitem{Jiang1}
\bibinfo{author}{Xiaoyi \surnamestart Jiang\surnameend} \&
  \bibinfo{author}{Horst \surnamestart Bunke\surnameend}
  (\bibinfo{year}{1996}): \emph{\bibinfo{title}{Including geometry in graph
  representations: A quadratic-time graph isomorphism algorithm and its
  applications}}.
\newblock In \bibinfo{editor}{Petra \surnamestart Perner\surnameend},
  \bibinfo{editor}{Patrick \surnamestart Wang\surnameend} \&
  \bibinfo{editor}{Azriel \surnamestart Rosenfeld\surnameend}, editors:
  {\slshape \bibinfo{booktitle}{Advances in Structural and Syntactical Pattern
  Recognition}}, \bibinfo{publisher}{Springer Berlin Heidelberg},
  \bibinfo{address}{Berlin, Heidelberg}, pp. \bibinfo{pages}{110--119},
  \doi{10.1007/3-540-61577-6_12}.

\bibitemdeclare{inproceedings}{Jiang2}
\bibitem{Jiang2}
\bibinfo{author}{Xiaoyi \surnamestart Jiang\surnameend} \&
  \bibinfo{author}{Horst \surnamestart Bunke\surnameend}
  (\bibinfo{year}{1998}): \emph{\bibinfo{title}{Marked subgraph isomorphism of
  ordered graphs}}.
\newblock In \bibinfo{editor}{Adnan \surnamestart Amin\surnameend},
  \bibinfo{editor}{Dov \surnamestart Dori\surnameend}, \bibinfo{editor}{Pavel
  \surnamestart Pudil\surnameend} \& \bibinfo{editor}{Herbert \surnamestart
  Freeman\surnameend}, editors: {\slshape \bibinfo{booktitle}{Advances in
  Pattern Recognition}}, \bibinfo{publisher}{Springer Berlin Heidelberg},
  \bibinfo{address}{Berlin, Heidelberg}, pp. \bibinfo{pages}{122--131},
  \doi{10.1007/BFb0033230}.

\bibitemdeclare{book}{taocpIII}
\bibitem{taocpIII}
\bibinfo{author}{Donald \surnamestart Knuth\surnameend} (\bibinfo{year}{1999}):
  \emph{\bibinfo{title}{The Art of Computer Programming: Sorting and Searching,
  Volume 3}}.
\newblock \bibinfo{publisher}{Addison-Wesley}, \bibinfo{address}{Reading MA}.

\bibitemdeclare{inproceedings}{mlir}
\bibitem{mlir}
\bibinfo{author}{Chris \surnamestart Lattner\surnameend},
  \bibinfo{author}{Mehdi \surnamestart Amini\surnameend}, \bibinfo{author}{Uday
  \surnamestart Bondhugula\surnameend}, \bibinfo{author}{Albert \surnamestart
  Cohen\surnameend}, \bibinfo{author}{Andy \surnamestart Davis\surnameend},
  \bibinfo{author}{Jacques \surnamestart Pienaar\surnameend},
  \bibinfo{author}{River \surnamestart Riddle\surnameend},
  \bibinfo{author}{Tatiana \surnamestart Shpeisman\surnameend},
  \bibinfo{author}{Nicolas \surnamestart Vasilache\surnameend} \&
  \bibinfo{author}{Oleksandr \surnamestart Zinenko\surnameend}
  (\bibinfo{year}{2021}): \emph{\bibinfo{title}{{{MLIR}}: Scaling Compiler
  Infrastructure for Domain Specific Computation}}.
\newblock In: {\slshape \bibinfo{booktitle}{2021 {{IEEE/ACM}} International
  Symposium on Code Generation and Optimization (CGO)}}, pp.
  \bibinfo{pages}{2--14}, \doi{10.1109/CGO51591.2021.9370308}.

\bibitemdeclare{article}{messmer}
\bibitem{messmer}
\bibinfo{author}{Bruno~T. \surnamestart Messmer\surnameend} \&
  \bibinfo{author}{Horst \surnamestart Bunke\surnameend}
  (\bibinfo{year}{1999}): \emph{\bibinfo{title}{A decision tree approach to
  graph and subgraph isomorphism detection}}.
\newblock {\slshape \bibinfo{journal}{Pattern Recognit.}} \bibinfo{volume}{32},
  pp. \bibinfo{pages}{1979--1998}, \doi{10.1016/S0031-3203(98)90142-X}.

\bibitemdeclare{inproceedings}{pytorch}
\bibitem{pytorch}
\bibinfo{author}{Adam \surnamestart Paszke\surnameend}, \bibinfo{author}{Sam
  \surnamestart Gross\surnameend}, \bibinfo{author}{Francisco \surnamestart
  Massa\surnameend}, \bibinfo{author}{A.~\surnamestart Lerer\surnameend},
  \bibinfo{author}{J.~\surnamestart Bradbury\surnameend},
  \bibinfo{author}{G.~\surnamestart Chanan\surnameend},
  \bibinfo{author}{T.~\surnamestart Killeen\surnameend},
  \bibinfo{author}{Z.~\surnamestart Lin\surnameend},
  \bibinfo{author}{N.~\surnamestart Gimelshein\surnameend},
  \bibinfo{author}{L.~\surnamestart Antiga\surnameend},
  \bibinfo{author}{A.~\surnamestart Desmaison\surnameend},
  \bibinfo{author}{A.~\surnamestart K{\"o}pf\surnameend},
  \bibinfo{author}{E.~\surnamestart Yang\surnameend},
  \bibinfo{author}{Z.~\surnamestart DeVito\surnameend},
  \bibinfo{author}{M.~\surnamestart Raison\surnameend},
  \bibinfo{author}{A.~\surnamestart Tejani\surnameend},
  \bibinfo{author}{S.~\surnamestart Chilamkurthy\surnameend},
  \bibinfo{author}{B.~\surnamestart Steiner\surnameend},
  \bibinfo{author}{L.~\surnamestart Fang\surnameend},
  \bibinfo{author}{J.~\surnamestart Bai\surnameend} \&
  \bibinfo{author}{S.~\surnamestart Chintala\surnameend}
  (\bibinfo{year}{2019}): \emph{\bibinfo{title}{{PyTorch}: An Imperative Style,
  High-Performance Deep Learning Library}}.
\newblock In: {\slshape \bibinfo{booktitle}{Neural Information Processing
  Systems}}, \doi{10.5555/3454287.3455008}.

\bibitemdeclare{article}{mqsubiso}
\bibitem{mqsubiso}
\bibinfo{author}{Xuguang \surnamestart Ren\surnameend} \&
  \bibinfo{author}{Junhu \surnamestart Wang\surnameend} (\bibinfo{year}{2016}):
  \emph{\bibinfo{title}{{Multi-Query Optimization for Subgraph Isomorphism
  Search}}}.
\newblock {\slshape \bibinfo{journal}{Proc. VLDB Endow.}}
  \bibinfo{volume}{10}(\bibinfo{number}{3}), p. \bibinfo{pages}{121–132},
  \doi{10.14778/3021924.3021929}.

\bibitemdeclare{article}{sellis}
\bibitem{sellis}
\bibinfo{author}{Timos~K. \surnamestart Sellis\surnameend}
  (\bibinfo{year}{1988}): \emph{\bibinfo{title}{{Multiple-Query
  Optimization}}}.
\newblock {\slshape \bibinfo{journal}{ACM Trans. Database Syst.}}
  \bibinfo{volume}{13}(\bibinfo{number}{1}), p. \bibinfo{pages}{23–52},
  \doi{10.1145/42201.42203}.

\bibitemdeclare{article}{TKET}
\bibitem{TKET}
\bibinfo{author}{Seyon \surnamestart Sivarajah\surnameend},
  \bibinfo{author}{Silas \surnamestart Dilkes\surnameend},
  \bibinfo{author}{Alexander \surnamestart Cowtan\surnameend},
  \bibinfo{author}{Will \surnamestart Simmons\surnameend},
  \bibinfo{author}{Alec \surnamestart Edgington\surnameend} \&
  \bibinfo{author}{Ross \surnamestart Duncan\surnameend}
  (\bibinfo{year}{2020}): \emph{\bibinfo{title}{tket: a retargetable compiler
  for {NISQ} devices}}.
\newblock {\slshape \bibinfo{journal}{Quantum Science and Technology}}
  \bibinfo{volume}{6}(\bibinfo{number}{1}), p. \bibinfo{pages}{014003},
  \doi{10.1088/2058-9565/ab8e92}.

\bibitemdeclare{book}{catalan}
\bibitem{catalan}
\bibinfo{author}{Richard~P. \surnamestart Stanley\surnameend}
  (\bibinfo{year}{2015}): \emph{\bibinfo{title}{Catalan Numbers}}.
\newblock \bibinfo{publisher}{Cambridge University Press},
  \doi{10.1017/CBO9781139871495}.

\bibitemdeclare{article}{qeso}
\bibitem{qeso}
\bibinfo{author}{Amanda \surnamestart Xu\surnameend}, \bibinfo{author}{Abtin
  \surnamestart Molavi\surnameend}, \bibinfo{author}{Lauren \surnamestart
  Pick\surnameend}, \bibinfo{author}{Swamit \surnamestart Tannu\surnameend} \&
  \bibinfo{author}{Aws \surnamestart Albarghouthi\surnameend}
  (\bibinfo{year}{2023}): \emph{\bibinfo{title}{Synthesizing Quantum-Circuit
  Optimizers}}.
\newblock {\slshape \bibinfo{journal}{Proc. ACM Program. Lang.}}
  \bibinfo{volume}{7}(\bibinfo{number}{PLDI}), \doi{10.1145/3591254}.

\bibitemdeclare{inproceedings}{quartz}
\bibitem{quartz}
\bibinfo{author}{Mingkuan \surnamestart Xu\surnameend}, \bibinfo{author}{Zikun
  \surnamestart Li\surnameend}, \bibinfo{author}{Oded \surnamestart
  Padon\surnameend}, \bibinfo{author}{S.~\surnamestart Lin\surnameend},
  \bibinfo{author}{J.~\surnamestart Pointing\surnameend},
  \bibinfo{author}{A.~\surnamestart Hirth\surnameend},
  \bibinfo{author}{H.~\surnamestart Ma\surnameend},
  \bibinfo{author}{J.~\surnamestart Palsberg\surnameend},
  \bibinfo{author}{A.~\surnamestart Aiken\surnameend}, \bibinfo{author}{U.A.
  \surnamestart Acar\surnameend} \& \bibinfo{author}{Z.~\surnamestart
  Jia\surnameend} (\bibinfo{year}{2022}): \emph{\bibinfo{title}{{Quartz:
  Superoptimization of Quantum Circuits}}}.
\newblock In: {\slshape \bibinfo{booktitle}{Proceedings of the 43rd ACM SIGPLAN
  International Conference on Programming Language Design and Implementation}},
  \bibinfo{series}{PLDI 2022}, \bibinfo{publisher}{Association for Computing
  Machinery}, \bibinfo{address}{New York, NY, USA}, p.
  \bibinfo{pages}{625–640}, \doi{10.1145/3519939.3523433}.

\end{thebibliography}

\appendix

\section{Proofs}
\subsection{Proof of \cref{prop:flatgraphwidth}}\label{app:flatgraphwidth}
\flatgraphwidth*
\begin{proof}
  For any acyclic linear path $P \subseteq E^\ast$ in $G$ consider its two endpoints $v_1$ and $v_2$, i.e. the two vertices in $G$
  that are only incident to one edge in $P$ (linear paths are never empty).
  Let $p_1$ and $p_2$ be the ports where the first and last edges are attached to $v_1$ and $v_2$ respectively.
  Let $p_i' \in ports(v_i)$ such that $p_i' \sim p_i$ for $i = 1, 2$.
  By \cref{eq:linpath}, either $\lambda(v_i, p_i') \in P$ or $\lambda(v_i, p_i') = \omega$.
  By injectivity of $\lambda(v_i, \cdot)$, the first case implies $p_i = p_i'$ as we would otherwise have two edges
  $\lambda(v_i, p_i') \neq \lambda(v_i, p_i)$ in $P$ incident to $v_i$.
  We thus conclude that either $\lambda(v_i, p_i') = \omega$ or $p_i$ is in a singleton equivalence class of $\sim$ in $ports(v_i)$.

  There are $n_\omega$ pairs $(v, p) \in V(G) \times \mathcal{P}$ such that
  $\lambda(v_i, p_i') = \omega$ and $n_{\textrm{odd}}$ pairs $(v, p) \in V(G) \times \mathcal{P}$ such that
  the equivalence class of $p$ is a singleton set in $ports(v)$.
  We conclude that there can be at most $n_\textrm{odd} + n_\omega$ end ports of linear paths
  in $G$. As every linear path has two end ports and every end port must be
  distinct, the result follows.
\end{proof}

\subsection{Proof of \cref{prop:cananchors}}\label{app:cananchors}
\cananchors*
\begin{proof}
Termination: we count the number of times \textsc{ConsumePath} is called
in one execution of \textsc{CanonicalAnchors}.
The call on line 3 happens exactly once, so we can ignore it.
On the other hand, the \texttt{path} argument to \textsc{ConsumePath} will always
be distinct between any two calls on line 33: it is either a path on a previously
unseen linear path, or it is a strict subset of the \texttt{path} argument passed
to the current call.
As there are $w$ linear path with at most $d$ vertices, there is a finite number of
calls to \textsc{ConsumePath}.
The while loop on lines 14--18 pops an element from the \texttt{path} queue
at each iteration, so can only be executed a finite number of times. Thus we can
conclude that the \textsc{ConsumePath} procedure always terminates.

Correctness: \textsc{CanonicalAnchors} returns $w$ vertices:
in every call to \textsc{ConsumePath}, the only non-recursive insertion to the
list of anchors is the initialisation of the \texttt{anchors} list on line 31.
This insertion happens if and only if \texttt{unseen} is non-empty (line 14).
Using the assumption that every vertex is on at most 2 linear paths, we can furthermore
restrict ourselves to $|\texttt{unseen}| = 1$.
Thus the size of \texttt{seen\_paths} increases by one at every recursion (line 21).
The size of \texttt{seen\_paths} is bounded by $w$ and thus $w$ anchors
are added to the \texttt{anchors} list over the execution of \textsc{CanonicalAnchors}.

Let $X$ be the vertices returned by \textsc{CanonicalAnchors} as a set.
It remains to be shown that the $X$-split graph of $G$ is connected and acyclic.
A cycle $C$ in split$_X(G)$ must have edges on at least 2 distinct linear paths by the
second assumption of \cref{sec:assumptions}. Say there are $k > 2$ distinct linear paths
on the cycle. Every anchor vertex on the cycle can be on either 1 or 2 linear paths (we
assumed in \cref{sec:assumptions} that no vertex is on more than two paths).
There must be at least $k$ anchor vertices on $C$ whose two adjacent edges
that are also in $C$ are on two different linear paths---one for every ``switch'' of linear path on $C$.
However, by line 14, for every anchor there is at least one unseen linear path, so for $k$
anchors there must be at least $k+1$ linear paths in $C$, which is impossible.

For connectedness, observe that for every vertex $v$
in the split graph there is a path from the root $r$ to $v$.
In $G$, such a path is obtained by following the graph traversal implicit in the calls to \textsc{ConsumePath}:
let $\tilde{v}$ be the vertex in $G$ that when split generates $v$.
Every vertex in $G$ appears in the \texttt{path} argument to \textsc{ConsumePath}
at least once.
There is thus an anchor $a \in X$ with a path along a linear path from $a$ to $\tilde{v}$.
Applying this argument recursively, 
there is a sequence of anchors $r = a_1, a_2, \dots, a_k = \tilde{v}$ corresponding to successive calls to \textsc{ConsumePath}
such that for all $1 \leq i < k$ there is a linear path between $a_i$ and $a_{i+1}$.

We show that the path from $r$ to $\tilde{v}$ through $a_1, \dots, a_k$
is mapped onto a path in the split graph.
In other words, we need to show that the edges along the path are rewired in such a way
that adjacent edges along the path are mapped to edges adjacent to the same split vertex.
We partition the path into sections from $a_i$ to $a_{i+1}$ for $i = 1, \dots, k-1$
and consider each subpath separately.
Let $e_1, \dots, e_m$ be the edges of the subpath from $a_i$ to $a_{i+1}$.
The first edge $e_1$ on this subpath is always in the split graph as $a_i \in X$ and thus
is not split.
Every other edge, on the other hand, is on the same linear path as $e_1$.
Thus for $1 \leq j < m$,  if $e_j$ ends in port $p$ and the next edge $e_{j+1}$ starts
in port $p'$, then $p \sim p'$. Thus both edges are mapped to the same split vertex,
concluding that the path from $r$ to $v$ is also a path in the split graph.

Complexity: Note that a recursive call to \textsc{ConsumePath} (line 33) occurs if and only if the current call is adding a new element to the list of anchor vertices (line 31). 
Since we have previously established that the number of anchor vertices returned by \textsc{CanonicalAnchors} is $w$, it follows that there are at most $w$ recursive calls to \textsc{ConsumePath}. 
Therefore, to prove the runtime complexity $O(w^2\cdot d)$ of \textsc{CanonicalAnchors}, it remains to show that the execution of the body of \textsc{ConsumePath}---excluding line 33---runs in $O(w\cdot d)$.
This is straightforward to check for all lines but 18 and 26.
Line 26 is executed at most $w$ times on a single call to \textsc{ConsumePath}, and since \texttt{lp.split_at(v)} simply needs to traverse a linear path of at most length $d$, the required runtime complexity holds.
On the other hand, line 18 is executed at most $d$ times on a single call to \textsc{ConsumePath}---since that is what it takes for line 17 to pop all elements from the path.

Assuming the list of linear paths was computed in advance (in time $O(w \cdot d)$,
before the first call to \textsc{ConsumePath} in line 3) and each linear path is given a unique index $0\dots w-1$,
we can store the \texttt{seen\_paths} set and the set of linear paths for each vertex as ordered lists of linear path indices.
The corresponding set \texttt{G.linear\_paths(v)} can be stored as an attribute of \texttt{v} and be retrieved in constant time (assuming
for instance that vertices are indexed from $0$ to $|G| - 1$).
Other than that, line 18 is a set operation that can be realised in a single $O(w)$ pass over
the ordered lists \texttt{unseen\_paths} and \texttt{G.linear\_paths(v)}
since both have at most $w$ elements.\qedhere
\end{proof}
\subsection{Proof of \cref{prop:ctstrings}}\label{app:ctstrings}
\ctstrings*

\begin{proof}
The $\Rightarrow$ direction is straightforward.
By assumption, the root $r_1$ in $T_1$ is mapped to the root $r_2$ in $T_2$.
Non-root anchors on the other hand are precisely the vertices on more than one path in $T_1$ and $T_2$.
If $\varphi: T_1 \to T_2$ is an injection of
trees of the same width, then all non-root anchors of $T_1$ must be mapped to the non-root
anchors of $T_2$.
Every linear path in $T_1$ includes at least one anchor vertex.
This must be mapped by $\varphi$ to a path in $T_2$,
through the corresponding anchor vertex in $T_2$.
As $\varphi$ preserves the port labels (\cref{eq:emb}), the image path in $T_2$ must be
a subpath of a linear path of $T_2$.
As the linear paths are split and ordered starting from anchor vertices, the string
encoding of every split linear path in $T_1$ will be a prefix of the string encodings
of split linear paths in $T_2$. 

$\Leftarrow$: it suffices to show that every path from root to a vertex in $T_1$
is also a path from root to a vertex in $T_2$ and that the vertex labels coincides.
A path $P$ from root $r_1$ in $T_1$ can be partitioned into a sequence of paths
$P = P_1 \cdots P_k$, which all
start at anchors and are subpaths of linear paths of $T_1$.
These subpaths corresponds to a sequence of prefixes of $s_{\alpha_1}, s_{\alpha_k}$ in the string encoding,
which are also prefixes of $t_{\alpha_1}, \dots, t_{\alpha_k}$. 
Since the vertex labels are stored in the tuple string encoding, we know that
the end vertex of the path $P_1\cdots P_i$ coincides with the anchor of $t_{\alpha_{i+1}}$
in $T_2$.
Applying this argument recursively on the chain of linear subpaths, we conclude
that $P = P_1 \cdots P_k$ is also a path in $T_2$.
Finally, the vertex labels must coincide on the shared domain of definition, as
the string encoding coincide.
\Cref{eq:emb} can be shown to be satisfied using a similar argument to the one presented
in the proof of \cref{prop:treeinc}.
\end{proof}

\subsection{Proof of \cref{prop:allanchors}}\label{app:allanchors}
\allanchors*
\begin{proof}
Let $H \subseteq G$ be a connected subgraph of $G$ of width $w$.
We prove inductively over $w$ that 
\begin{equation}\label{eq:induchypo}
  \textsc{ConsumePath}(H, \texttt{path}, \texttt{seen\_paths}) \in
  \textsc{AllConsumePath}(G, w, \texttt{path}, \texttt{seen\_paths})
\end{equation}
for all arguments \texttt{path} and \texttt{seen\_paths}.
The statement in the proposition follows from this claim directly.

For the base case $w = 1$, 
\textsc{ConsumePath} will return \texttt{[new\_anchor]}, where \texttt{new\_anchor}
is obtained from lines 16--20 of \cref{lst:anchors}:
there is only one linear path and thus for every recursive call to \textsc{ConsumePath},
\texttt{unseen} will be empty, until \texttt{path} has been exhausted and the empty list is returned.
The definition of \texttt{new\_anchor} coincides with the one obtained from lines 20--28
of \cref{lst:extract}.
The only values of \texttt{w1}, \texttt{w2} and \texttt{w3} that satisfy the loop condition
on line 42 of \cref{lst:extract}
for $w = 1$ are $\texttt{w1} = \texttt{w2} = \texttt{w3} = 0$.
Using the base condition on lines 18--20 of \cref{lst:extract}, we 
conclude that \textsc{AllConsumePath}$(G, 1, \texttt{path}, \texttt{seen\_paths})$
returns \texttt{[[new\_anchor]]}, satisfying \cref{eq:induchypo}.

We now prove the claim for $w > 1$ by induction.
Using our simplifying assumptions, we obtain the assertion on line 32 of
\cref{lst:extract}, as documented.
For \cref{lst:anchors}, this assumption simplifies the loop on lines 34--37 to at most three
calls to \textsc{ConsumePath}
with arguments
$(H, P_{curr}, S_{curr})$,
$(H, P_\ell, S_\ell)$ and $(H, P_r, S_r)$ respectively,
where
\begin{itemize}
  \item 
  $P_{curr}$ is the value of the \texttt{path} variable after line 20,
  \item
  $P_\ell$ and $P_r$ refer to the two halves of the new linear path,
  as computed and stored in the variables \texttt{left\_path} and \texttt{right\_path} on line 28, and
  \item
  $S_{curr}, S_\ell$ and $S_r$ are the values of the \texttt{seen\_paths} variable
  after the successive updates on line 23 and two iterations of line 37.
\end{itemize}
Consider a call to \textsc{ConsumePath} (\cref{lst:anchors}) with arguments \texttt{G} = $H$
and some variables \texttt{path} and \texttt{seen\_paths}.
Let $w_{curr}, w_\ell$ and $w_r$ be the length of the values returned 
by the three recursive calls to \textsc{ConsumePath} of line 35.
As every anchor vertex reduces the number of unseen linear paths by exactly one
(using the simplifying assumptions),
it must hold that $w_{curr} + w_\ell + w_r + 1 = w$.
Thus for a call to \textsc{AllConsumePath} (\cref{lst:extract}) with arguments
\texttt{G} = $G$, \texttt{w} = $w$ and the same values for \texttt{path} and \texttt{seen\_paths},
there is an iteration of the \texttt{for} loop on line 42 of \cref{lst:extract}
such that $\texttt{w1} = w_{curr}, \texttt{w2} = w_\ell$ and $\texttt{w3} = w_r$.
The definition of \texttt{seen0} on line 38 of \cref{lst:extract} coincides with
the update to $\texttt{seen\_paths}$ on line 23 of \cref{lst:anchors};
it follows that on line 43 of \cref{lst:extract} the recursive call
$\textsc{AllConsumePath}(G, w_{curr}, P_{curr}, S_{curr})$ is executed.
From the induction hypothesis we obtain
that there is an iteration of the \texttt{for} loop on line 43 of \cref{lst:extract}
in which \texttt{anchors1} and \texttt{seen1} coincide with the
\texttt{new\_anchors} and \texttt{new\_seen\_paths} variables of the first iteration
of the \texttt{for} loop on line 34 of \cref{lst:anchors}.
In particular the value of \texttt{seen1} is equal $S_\ell$.

Repeating the argument, we obtain that there are iterations of the \texttt{for} loops
on lines 44 and 45 of \cref{lst:extract} that correspond to the second and third
calls to \textsc{ConsumePath} on line 35 of \cref{lst:anchors}.
Finally, the concatenation of anchor lists on line 47 of \cref{lst:extract} is equivalent
to the repeated concatenations on line 36 of \cref{lst:anchors} and so
we conclude that \cref{eq:induchypo} holds for $w$.
\end{proof}

\subsection{Proof of \cref{prop:allanchorscoro}}\label{app:allanchorscoro}
\allanchorscoro*
\begin{proof}
We restrict \texttt{split\_at} on line 37 to only return the first $d$
vertices on the linear path in each direction: vertices more than distance $d$ away
from the anchor cannot be part of a pattern of depth $d$.

We use the bound on the length of the list returned by calls to \textsc{AllConsumePath} of \cref{prop:catalanbound}
to bound the runtime.
We can ignore the non-constant runtime of the concatenation of the outputs of recursive calls
on line 47,
as the total size of the outputs is asymptotically at worst of the same complexity as the
runtime of the recursive calls themselves.
Excluding the recursive calls, the only remaining lines of \textsc{AllConsumePath}
that are not executed in constant time are
the \texttt{while} loop on lines 24--28 and the \texttt{split\_at} call on line 37.

Consider the recursion tree of \textsc{AllConsumePath}, i.e. the tree in which
the nodes are the recursive calls to \textsc{AllConsumePath} and the children
are the executions spawned by the nested \texttt{for} loops on line 42--48.
This tree has at most
\[C_w = \Theta\left(\frac{c^w}{w^{\sfrac32}}\right)\] leaves.
A path from the root to a leaf corresponds to a stack of recursive calls
to \textsc{AllConsumePath}.
Along this recursion path, the \texttt{seen\_paths} set is
always strictly growing (line 38) and the vertices popped from the \texttt{path}
queue on line 27 are all distinct.
\texttt{split\_at} is called once for each of the $w$ linear path that
are added to \texttt{seen\_paths}.
For each linear path two paths of length at most $d$ are traversed and returned.
Thus the total runtime of \texttt{split\_at} along a path from root to leaf
in the recursion tree is in $O(w \cdot d)$.
Similarly, the number of executions of the lines 25--28 is bound by the number
of elements that were added to a \texttt{path} queue,
as for every iteration an element is popped off the queue on line 27.
This is equal to the number of elements returned by \texttt{split\_at}, resulting
in the same complexity.
We can thus bound the overall complexity of executing the entire recursion tree
by $O(C_w \cdot w \cdot d) = O(\frac{c^w \cdot d}{w^{\sfrac12}})$.
\end{proof}

\section{Lower bound on the number of patterns}\label{app:proofellbound}
\begin{prop}\label{prop:ellbound}
    Let $N_{w,d}$ be the number of port graphs of width $w$, depth $d$ and maximum degree $\Delta \geq 4$.
    We can lower bound
    \[N_{w,d} > \left(\frac{w}{2e}\right)^{\Theta(wd)},\]
    assuming $w \leq o(2^d)$.
\end{prop}
\noindent
In the regime of interest, $w$ is small, so the assumption $w \leq  o(2^d)$ is not a restriction.
In the main text we use the bound $|P| \leq w\cdot d$ to avoid introducing the circuit depth.
The bound stated in \cref{eq:regime} is thus slightly looser.

\begin{proof}
Let $w, d > 0$ and $\Delta \geq 4$ be integers.
We wish to lower bound the number of port graphs of depth $d$, width $w$ and maximum degree $\Delta$.
It is sufficient to consider a restricted subset of such port graphs, whose size can be easily lower bounded.
We will count a subset of CX quantum circuits, i.e. circuits with only $CX$ gates, a two-qubit
non-symmetric gate.
Because we are using a single gate type, this is equivalent to counting a subset of port graphs with
vertices of degree 4.
Assume w.l.o.g that $w$ is a power of two.
We consider CX circuits constructed from two circuits with $w$ qubits composed in sequence:
\begin{itemize}
    \item \textbf{Fixed tree circuit}: A $\log_2(w)$-depth circuit that connects qubits pairwise in such a way that the resulting
    port graph is connected. We fix such a tree-like circuit and use the same circuit for all CX circuits.
    We can use this common structure to fix an ordering of the $w$ qubits, that refer to as qubits
    $1,\dots,w$.
    \item \textbf{Bipartite circuit}: A CX circuit of depth $D = d - \log_2(w)$ with exactly $\sfrac{w}{2}\cdot (d - \log_2(w))$
    CX gates, each gate acting on a qubit $1 \leq q_1 \leq \sfrac{w}{2}$ and a qubit
    $\sfrac{w}{2} < q_2 \leq w$.
\end{itemize}
The following circuit illustrates the construction:
\begin{center}
\begin{tikzpicture}
\node at (-2, 0){};
\node at (0.5, 0){Fixed tree circuit};
\node at (4.7, 0) {Variable bipartite circuit};
\node at (7, 0){};
\end{tikzpicture}
\begin{quantikz}[
  wire types={n,q,q,q,q,n,q,q,q,q},
  row sep=0.3cm,
  column sep=0.5cm,
]
  &\gate[style={
    minimum width=1.7cm,
    decorate,decoration={brace,amplitude=5pt,raise=-1ex}
  }]{}&&&&&\gate[style={
    minimum width=1.5cm,
    decorate,decoration={brace,amplitude=5pt,raise=-1ex}}]{}
  \\[2mm]
  &\ctrl{1} & \ctrl{3} & \ \ldots\ & \ctrl{4}\slice[style=black]{} &\ctrl{4}&&&&\ \cdots \\
  &\targ{} && \ \ldots\ &&&\ctrl{3}&&&\ \cdots\\
  &\ctrl{1} && \ \ldots\ &&&&\ctrl{2}&&\ \cdots \\
  &\targ{} &\targ{}&\ \ldots\ &&&&&\ctrl{1}&\ \cdots\\[-1mm]
  &\rotatebox{90}{$\cdots$}&&&\wire[d][4]{q}\rotatebox{90}{$\cdots$}&\wire[d][1]{q}\rotatebox{90}{$\cdots$}&\wire[d][2]{q}\rotatebox{90}{$\cdots$}&\wire[d][3]{q}\rotatebox{90}{$\cdots$}&\wire[d][4]{q}\rotatebox{90}{$\cdots$}&\\[-1mm]
  &\ctrl{1} & \ctrl{3} &\ \ldots\ & &\targ{} &&&&\ \cdots \\
  &\targ{} &&\ \ldots\ &&& \targ{} &&&\ \cdots \\
  &\ctrl{1} &&\ \ldots\ &&&& \targ{} &&\ \cdots  \\
  &\targ{} &\targ{}&\ \ldots\ &\targ{}&&&&\targ{}&\ \cdots
\end{quantikz}
\end{center}
All that remains is to count the number of such bipartite circuits.
Every slice of depth 1 must have $w / 2$ CX gates acting on distinct qubits.
Every qubit $1$ to $w/2$ must interact with one of the qubits $w/2+1$ to $w$,
so there are $(w/2)!$ such depth 1 slices.
Repeating this depth 1 construction $D$ times and using
Sterling's approximation, we obtain a lower bound for the number of port graphs of depth $d$,
width $w$ and maximum degree at least 4:
\[
\left(\left(\frac{w}2\right)!\right)^D > \sqrt{w\pi}\left(\frac{w}{2e}\right)^{wD/2}
= \left(\frac{w}{2e}\right)^{\Theta(w\cdot d)}
\]
where we used $w = o(2^d)$ to obtain $\Theta(D) = \Theta(d)$ in the last step.
\end{proof}

\section{Quantum circuits as port graphs}\label{app:qc-pg}
A relevant consideration when viewing quantum circuits as port graphs is 
the question of equality on circuits.
We consider two circuits to be equal if they are equal as port graphs.
This sense of equality is more general than equality of ordered lists of gates, another
common internal representation of quantum computations,
but does not account for commuting gates or gate \emph{symmetries}.
An example of a symmetric gate type is the CZ gate, a gate type of arity $n=2$
that is symmetric in its arguments
\begin{center}
\begin{quantikz}
&\ctrl{1} & \\
& \control{} &
\end{quantikz}=\begin{quantikz}
&\permute{2,1}&\ctrl{1} & \permute{2,1}& \\
& &\control{} &&
\end{quantikz}
\end{center}
that is to say, exchanging the order of the inputs and outputs does not
change the computation.
Viewed as port graphs, however, the left and right hand
side are distinct circuits
\begin{center}
    \begin{quantikz}
        &\gate[2]{CZ}\gateinput{0}\gateoutput{0} & \\
        & \gateinput{1}\gateoutput{1} &
    \end{quantikz}$\neq$\begin{quantikz}
        &\permute{2,1}&\gate[2]{CZ}\gateinput{0}\gateoutput{0} & \permute{2,1}&\\
        && \gateinput{1}\gateoutput{1} &&
    \end{quantikz}$=$\begin{quantikz}
        &\gate[2]{CZ}\gateinput{1}\gateoutput{1} & \\
        & \gateinput{0}\gateoutput{0} &
    \end{quantikz}
\end{center}
In the case that such symmetries need to be taken into account for pattern matching, there are
two simple solutions.
For rewriting purposes, one may choose to add a single rewrite rule
to express the symmetry explicitly, stating that the symmetric gate can be rewritten to itself
with the edge order reversed.
This will recover the full expressivity of the rewrite rule set, at the expense of additional rewrite
rule applications.

Alternatively, all instances of a pattern that are equivalent up to gate symmetries can be
enumerated and added as separate patterns to the matcher.
This approach is particularly appealing
as the runtime of the pattern matcher will remain unchanged, 
despite the increase in the number of patterns (exponential in the number of symmetric gates).
The trade-off is increased pre-compilation time and pattern matcher size.

\section{Properties of the Canonical Tree representation}\label{app:ctrepr}
We provide here the exact derivations of the properties of the CT
representation that we rely on, namely an injective map from the port graph representation
to CTs, invariance of the CT representation under pattern embeddings
and the string encoding of CT trees.

\paragraph{Equivalence of the CT representation.}
A connected port graph $G$ is fully defined by the set of edges, given as a set of pairs in $V \times \mathcal{P}$.
Given the CT representation of $G$ with vertices $\tilde{V}$,
alongside a map $merge: \tilde{V} \to V$ that maps the vertices of the CT
representation to the vertices of $G$, it is immediate that $G$ can be recovered by mapping every $(v, p) \in \tilde{V}\times \mathcal{P}$
to $(merge(v), p) \in V \times \mathcal P$.

Up to isomorphism in the co-domain $V$, we can store $merge$ by storing the partition of $\tilde{V}$ into sets with the same image.
We introduce for this a map $\tilde{V} \to \tilde{V}$ that maps
every vertex to a canonical representative of the partition---for instance the vertex closest
to the root in CT.
This map can be stored as vertex labels of CT, which we
can refer to as the labelled CT representation for distinction.
However in the main text, it is always the labelled representation that
is meant when CT representations are discussed.

We thus have a bijective map between the labelled CT representation of $G$ and the port graph $G$.
Furthermore, this map preserves the linear paths, i.e. it maps one to one the linear paths of the
labelled CT representation to the linear paths of $G$.
For all purposes, we can thus treat the labelled CT representation as an equivalent
representation of $G$.

\paragraph{Invariance under pattern embedding.}
Unlike graphs, rooted trees can be defined in a way that is invariant under
bijective relabelling of the vertices by using the invariant port labels.
Every tree vertex is either the root vertex or it is uniquely identified by the path to it
from the root. Since paths can be defined in terms of port labels, paths are invariant
under pattern embeddings of the underlying graph.

For trees $T$ and $T'$, let $S$ and $S'$ be the sets of their respective vertices expressed
as sequences of port labels.
We thus define tree inclusion and equality only up to vertex relabelling: $T \subseteq T'$ if and only if $S \subseteq S'$, and $T = T'$ if and only if $S = S'$.
On labelled trees, we also require inclusion (resp. equality) of the vertex label maps, including in particular the $merge$ map of labelled CT representations.
As a result, subtree relations in labelled CT representations correspond to subgraphs
of the original graph, in effect reducing the pattern matching problem
on port graphs to a problem of tree inclusion on CT representations.
This statement is formalised in \cref{prop:treeinc}.

\paragraph{String encoding of CT representations.}
In order for our string encoding of CT representations
to map tree inclusion to string prefixes,
we recall that the anchor set
in \cref{eq:treeinc} is fixed: a subtree of $T$ with the same anchor set
can only be obtained by shortening the linear paths at their ends---
the resulting subpath will always contain the anchor vertex.
Given a linear path $L$ of $T$, we thus split $L$ at the anchor on $L$ and obtain
two paths $L_1, L_2$ starting from the anchor to the ends of $L$. 
For any subtree $T' \subseteq T$, the linear path $L'$ that is a subpath of $L$
will split into $L_1', L_2'$, prefixes of $L_1$ and $L_2$ respectively.

With an appropriate string representation of CT vertices and their labels,
this will encode all linear paths.
In the same way that the $merge$ map of the labelled CT representation is used to
restore the original graph from the split CT vertices, we use it to recover the anchor
vertices from the split linear paths.
Finally, to order the linear paths in the string tuple, we use for instance
the order of their anchors induced by port ordering.


\section{Prefix Trees}\label{app:prefixmatch}
Our main result is achieved by reducing a tree inclusion problem to the following problem.
\paragraph{String prefix matching.}
Consider the following computational problem over strings.
Let $\Sigma$ be a finite alphabet and consider $\mathcal{W} = (\Sigma^*)^w$
the set of $w$-tuples of strings over $\Sigma$.
For a string tuple $(s_1, \dots, s_w) \in \mathcal{W}$ and a set of string tuples $\mathcal{D} \subseteq \mathcal{W}$,
the $w$-dimensional string prefix matching consists in finding the set
\[
    \{ (p_1, \dots, p_w) \in \mathcal{D} \ | \ \text{for all }1 \leq i \leq w: p_i\text{ is a prefix of }s_i \}.
\]
This string problem can be solved using a $w$-dimensional prefix tree.
We give a short introduction to prefix trees for the string case but refer
to standard literature for more details~\cite{taocpIII}.

\paragraph{One-dimensional prefix tree.}
Let $P_1, \dots, P_\ell \in \mathcal{A}^\ast$ be strings on some alphabet $\mathcal{A}$.
Given an input string $s\in\mathcal{A}^\ast$, we wish to find the set of
patterns $\{ P_{1 \leq i \leq \ell} | P_i \subseteq s\}$, i.e. $P_i$ is a prefix of $s$.

The prefix tree of $P_1, \dots, P_\ell$ is a tree with a tree node for each prefix of
a pattern. The children of an internal node are the strings that extend the prefix
by one character. The root of the tree is the empty string.
Each tree node also stores a list of matching patterns, with each pattern stored in the unique corresponding node.
Every prefix tree has an empty string node, which is the root of the tree.
For every inserted pattern of length at most $L$ nodes are inserted, one
for every non-empty prefix of the pattern. Thus a one-dimensional prefix tree
has at most $\ell \cdot L + 1$ nodes and can be constructed in time $O(\ell \cdot L)$.

Given an input $s \in \mathcal{A}^\ast$, we can find the set of matching patterns
by traversing the prefix tree of $P_1, \dots, P_\ell$ starting from the root.
We report the list of matching patterns at the current node
and move to the child node that is still a prefix of $s$, if it exists.
This procedure continues until no more such child exists.
In total the traversal takes time $O(|s|)$, as every character of $s$ is visited
at most once.

Note that in theory the number of reported pattern matches can dominate the runtime
of the algorithm. We can avoid this
by returning the list of matches as an iterator, stored as a list of pointers
to the tree nodes matching lists.
\paragraph{Multi-dimensional prefix tree.}
A $w$-dimensional prefix tree for $w > 1$ is defined recursively as a one-dimensional
prefix tree that at each node stores a $w-1$-dimensional prefix tree.
Given an input $w$-tuple $(s_1, \dots, s_w) \in (\mathcal{A}^\ast)^w$,
the traversal of the $w$-dimensional prefix tree is done by traversing the one-dimensional
prefix tree on the input $s_1$ until no child is a prefix of the input,
and then recursively traversing the $w-1$-dimensional prefix tree on $(s_2, \dots, s_w)$.
Similarly to the one-dimensional case, the list of matching patterns is stored at prefix tree nodes
and reported during traversal.
The traversal thus takes time $O(|s_1| + \cdots + |s_w|)$, as every character of $s$ is visited
at most once.

For $\ell$ tuples of size $w$ of words of maximum length $L$, we can bound the number of nodes
of the $w$-dimensional prefix tree by $1 + (\ell \cdot L)^w$.
The runtime and space complexity of the construction of the $w$-dimensional prefix tree
is thus in $O((\ell \cdot L)^w)$, summarised in the result:

\begin{prop}\label{prop:prefixmatch}
    Let $\mathcal{D} \subseteq \mathcal{W}$ be a set of string tuples
    and $L$ the maximum length of a string in a tuple of $\mathcal{D}$.
    There is a prefix tree with at most $(\ell \cdot L)^w + 1$ nodes
    that encodes $\mathcal{D}$ that can be used to solve
    the $w$-dimensional string prefix matching problem
    in time $O(|s_1| + \cdots + |s_w|)$.
\end{prop}

\section{Open source implementation}\label{app:portmatching}
The code is available at \url{https://github.com/lmondada/portmatching/}.
All benchmarking can be reproduced using the tooling and instructions at \url{https://github.com/lmondada/portmatching-benchmarking}.

We represent all the pattern matching logic within a generalised finite state automaton, composed
of states and transitions.
This formalism is used to
traverse the graph input and express both
the prefix tree of the string prefix matching problem and the (implicit) recursion tree of \cref{lst:extract}
in \cref{sec:anchors}.
We sketch here the automaton definition. Further implementation details can be obtained from the \texttt{portmatching} project directly.

In the pre-computation step, the automaton is constructed based on the set of patterns
to be matched. It is then saved to the disk;
a run of the automaton on an input graph $G$
is the solution the pattern independent matching problem for the input $G$.
To run the automaton, we keep track of the set of current states, initialised to
a singleton root state and updated following allowed transitions
from one of the current states.
Which transitions are allowed is computed using predicates on the input graph
stored at the transitions.
This is repeated until no further allowed transitions exist from a current state.

At any one state of the automaton, zero, one or several transitions may be allowed
depending on the input graph.
As the automaton is run for a given input graph $G$, we keep track of the vertices that
have been matched by the automaton so far with an injective map between a set
of unique symbols and the vertices of $G$.
Vertices in this map are the known vertices of $G$.
There are three main types of transitions:
\begin{itemize}
  \item A \textbf{constraint} transition asserts that a property of the known
  vertices holds. This can be checking for a vertex or edge label, or checking
  that an edge between two known vertices and ports exists.
  \item A \textbf{new vertex} transition asserts that there is an edge between
  a known vertex $v$ at a port $p$ and a new vertex at a port $p'$.
  The new vertex must not be any of the known vertices. When the transition
  is followed, a new symbol is introduced and the vertex is added to the symbol
  vertex map.
  \item A \textbf{set anchor} transition is an $\epsilon$-transition, i.e. a non-deterministic
  transition that is always allowed. Semantically, it designates a known vertex as an
  anchor.
\end{itemize}
By requiring that all constraint transitions from a given state assert mutually
exclusive predicates (such as edges starting from a given vertex and port, or
the vertex label of a given vertex), we can ensure that constraint transitions
are always deterministic.
New vertex transitions are also deterministic in finite depth patterns~\footnote{%
In cyclic and non-convex cases, it can happen that a vertex is both a known vertex of a large %
pattern and a new vertex within a smaller subpattern.}, so that
in the regime explored in this paper, the only source of non-determinism
is the choice of anchors.
Intuitively, this corresponds to the facts that the prefix tree traversal
of \cref{sec:toy} is deterministic while the anchors enumeration of \cref{lst:extract}
returns a multitude of options to be explored exhaustively.

To obtain a set of matching patterns from a run of the automaton, we store pattern matches
as lists at the automaton states.
When a state is added to the set of current states, its list of matches are added to the output.
To build the automaton, we consider one pattern at a time, convert it into a chain of
transitions of the above types that is then added to the state transition graph.
At the target state of the last
transition, we then add the pattern ID to the list of matched patterns.
\end{document}